\documentclass[sigplan,screen,authorversion,nonacm]{acmart}
\ccsdesc[500]{Theory of computation~Algebraic semantics}
\ccsdesc[500]{Theory of computation~Categorical semantics}
\ccsdesc[500]{Theory of computation~Logic and verification}
\ccsdesc[300]{Theory of computation~Type theory}

\usepackage{csquotes}
\usepackage[capitalise]{cleveref}
\usepackage{xspace}
\usepackage{xcolor}
\usepackage[color]{coqdoc}
\usepackage{url}
\usepackage{textgreek}
\usepackage{stmaryrd}
\usepackage{pmboxdraw}
\usepackage[pdf,all,2cell]{xy}\SelectTips{cm}{10}\SilentMatrices\UseAllTwocells
\usepackage{tikz}
\usetikzlibrary{arrows,matrix,decorations.pathmorphing,
  decorations.markings, calc, backgrounds}
\usepackage{listings}
\usepackage{xifthen} 
\usepackage{newunicodechar}
\usepackage[inline]{enumitem}
\usepackage{microtype}

\AtEndPreamble{%
\theoremstyle{acmdefinition}
\newtheorem{remark}[theorem]{Remark}}

\newunicodechar{⦃}{\ensuremath{\llparenthesis}}
\newunicodechar{⦄}{\ensuremath{\rrparenthesis}}
\newunicodechar{⨿}{\ensuremath{\amalg}}

\newunicodechar{●}{\ensuremath{\bullet}}

\newunicodechar{◾}{\ensuremath{\blacksquare}}
\newunicodechar{◽}{\ensuremath{\square}}
\newunicodechar{⟲}{\ensuremath{\circlearrowleft}}
\newunicodechar{⟳}{\ensuremath{\circlearrowright}}
\newunicodechar{▭}{\ensuremath{\boxdot}}
\newunicodechar{╝}{{\textSFxxvi}}
\newunicodechar{∑}{\ensuremath{\sum}}
\newunicodechar{⟦}{\ensuremath{\llbracket}}
\newunicodechar{⟧}{\ensuremath{\rrbracket}}
\newunicodechar{⇒}{\ensuremath{\Rightarrow}}

\newcommand{\naturalto}{%
  \mathrel{\vbox{\offinterlineskip
    \mathsurround=0pt
    \ialign{\hfil##\hfil\cr
      \normalfont\scalebox{1.2}{.}\cr
      $\longrightarrow$\cr}
  }}%
}

\newcommand{\teletype}[1]{\ensuremath{\mathtt{#1}}}
\newcommand{\systemname}[1]{\teletype{\color{darkgray}#1}\xspace}

\newcommand{\Agda}{\systemname{Agda}}
\newcommand{\UniMath}{\systemname{UniMath}}

\newcommand{\Coq}{\systemname{Coq}}

\newcommand{\Haskell}{\systemname{Haskell}}

\newcommand{\cfont}[1]{\mathsf{#1}}

\newcommand{\A}{\mathcal{A}}
\newcommand{\B}{\mathcal{B}}
\newcommand{\CC}{\mathcal{C}}
\newcommand{\D}{\mathcal{D}}

\newcommand{\V}{\mathcal{V}}
\newcommand{\T}{\mathcal{T}}

\newcommand{\ob}[1]{\ensuremath{{#1}_0}}
\newcommand{\constfunctor}[1]{\ensuremath{\underline{#1}}}
\newcommand{\Id}[1]{\cfont{Id}}
\newcommand{\Ptd}{\mathrm{Ptd}}
\newcommand{\idwt}[1]{\cfont{id}}
\newcommand{\vcomp}{\mathbin{\circ}} 
\newcommand{\hcomp}{\mathbin{\cdot}} 
\newcommand{\fat}[1]{\textbf{#1}}

\newcommand{\lam}{\cfont{lam}}
\newcommand{\app}{\cfont{app}}

\newcommand{\join}{\cfont{join}}
\newcommand{\prejoin}{\cfont{j}}

\newcommand{\set}{\ensuremath{\mathsf{Set}}\xspace}

\newcommand{\arity}{\cfont{ar}}
\newcommand{\lst}{\cfont{list}}
\newcommand{\sort}{\cfont{S}}
\newcommand{\ccsort}{\CC^{\sort}}

\newcommand{\ifte}[3]{\cfont{if}\,#1\,\cfont{then}\,#2\,\cfont{else}\,#3}
\newcommand{\tinl}{\cfont{inl}}
\newcommand{\tinr}{\cfont{inr}}

\newcommand{\proj}[1]{\cfont{pr}_{#1}}
\newcommand{\UU}{\cfont{U}}
\newcommand{\option}{\cfont{option}}
\newcommand{\arr}{\rightarrow}
\renewcommand{\th}[1]{\theta_{#1}}

\newcommand{\et}[1]{\eta_{#1}}              
\newcommand{\ta}[1]{\tau_{#1}}              
\newcommand{\gst}[2]{\llparenthesis {#2} \rrparenthesis}

\newcommand{\convert}{\ensuremath{\equiv}}
\newcommand{\eqdef}{\ensuremath{:\convert}}
\newcommand{\eqdec}{\ensuremath{=_{\mathrm{dec}}}}
\newcommand{\Endo}{\mathrm{Endo}}
\newcommand{\CAT}{\mathrm{CAT}}

\makeatletter
\def\prd#1{\@ifnextchar\bgroup{\prd@parens{#1}}{%
    \@ifnextchar\sm{\prd@parens{#1}\@eatsm}{%
    \@ifnextchar\prd{\prd@parens{#1}\@eatprd}{%
    \@ifnextchar\;{\prd@parens{#1}\@eatsemicolonspace}{%
    \@ifnextchar\\{\prd@parens{#1}\@eatlinebreak}{%
    \@ifnextchar\narrowbreak{\prd@parens{#1}\@eatnarrowbreak}{%
      \prd@noparens{#1}}}}}}}}
\def\prd@parens#1{\@ifnextchar\bgroup%
  {\mathchoice{\@dprd{#1}}{\@tprd{#1}}{\@tprd{#1}}{\@tprd{#1}}\prd@parens}%
  {\@ifnextchar\sm%
    {\mathchoice{\@dprd{#1}}{\@tprd{#1}}{\@tprd{#1}}{\@tprd{#1}}\@eatsm}%
    {\mathchoice{\@dprd{#1}}{\@tprd{#1}}{\@tprd{#1}}{\@tprd{#1}}}}}
\def\@eatsm\sm{\sm@parens}
\def\prd@noparens#1{\mathchoice{\@dprd@noparens{#1}}{\@tprd{#1}}{\@tprd{#1}}{\@tprd{#1}}}
\def\lprd#1{\@ifnextchar\bgroup{\@lprd{#1}\lprd}{\@@lprd{#1}}}
\def\@lprd#1{\mathchoice{{\textstyle\prod}}{\prod}{\prod}{\prod}({\textstyle #1})\;}
\def\@@lprd#1{\mathchoice{{\textstyle\prod}}{\prod}{\prod}{\prod}({\textstyle #1}),\ }
\def\tprd#1{\@tprd{#1}\@ifnextchar\bgroup{\tprd}{}}
\def\@tprd#1{\mathchoice{{\textstyle\prod_{(#1)}}}{\prod_{(#1)}}{\prod_{(#1)}}{\prod_{(#1)}}}
\def\dprd#1{\@dprd{#1}\@ifnextchar\bgroup{\dprd}{}}
\def\@dprd#1{\prod_{(#1)}\,}
\def\@dprd@noparens#1{\prod_{#1}\,}

\def\@eatnarrowbreak\narrowbreak{%
  \@ifnextchar\prd{\narrowbreak\@eatprd}{%
    \@ifnextchar\sm{\narrowbreak\@eatsm}{%
      \narrowbreak}}}
\def\@eatlinebreak\\{%
  \@ifnextchar\prd{\\\@eatprd}{%
    \@ifnextchar\sm{\\\@eatsm}{%
      \\}}}
\def\@eatsemicolonspace\;{%
  \@ifnextchar\prd{\;\@eatprd}{%
    \@ifnextchar\sm{\;\@eatsm}{%
      \;}}}

\newcommand{\longhash}{c26d11bef0b4d4c226ca854d8de340d684c2a10a}
\newcommand{\shorthash}{c26d11b}

\newcommand{\coqdocbasebaseurl}{https://unimath.github.io/doc/UniMath/\shorthash}

\newcommand{\coqdocbaseurl}{\coqdocbasebaseurl /UniMath.}
\newcommand{\urlhash}{\#}
\newcommand{\coqdocurl}[2]{\coqdocbaseurl #1.html\urlhash #2}
\newcommand{\nolinkcoqident}[1]{\nolinkurl{#1}} 
\makeatletter
\newcommand{\coqident}{\begingroup\@makeother\#\@coqident}
\newcommand{\@coqident}[3][]{%
  \ifthenelse{\isempty{#2}}%
  {\nolinkcoqident{#3}}%
  {{\color{blue}
    \ifthenelse{\isempty{#1}}%
  {\href{\coqdocurl{#2}{#3}}{\nolinkcoqident{#3}}}%
  {\href{\coqdocurl{#2}{#3}}{\nolinkcoqident{#1}}}}}%
\endgroup}
\newcommand{\coqfile}[2]{%
  {\color{blue}
  \ifthenelse{\isempty{#1}}%
  {\href{\coqdocbaseurl #2.html}{\nolinkcoqident{#2.v}}}%
  {\href{\coqdocbaseurl #1.#2.html}{\nolinkcoqident{#2.v}}}}}
\makeatother

\newcommand{\eg}{e.\,g.\xspace}
\newcommand{\aka}{a.\,k.\,a.\xspace}
\newcommand{\omegacocont}{$\omega$-co\-con\-tin\-u\-ous\xspace}
\newcommand{\omegacocontinuity}{$\omega$-co\-con\-ti\-nu\-ity\xspace}  
\newcommand{\ie}{i.\,e.\xspace}

\setcopyright{rightsretained}
\acmPrice{}
\acmDOI{10.1145/3497775.3503678}
\acmYear{2022}
\copyrightyear{2022}
\acmSubmissionID{poplws22cppmain-p10-p}
\acmISBN{978-1-4503-9182-5/22/01}
\acmConference[CPP '22]{Proceedings of the 11th ACM SIGPLAN International Conference on Certified Programs and Proofs}{January 17--18, 2022}{Philadelphia, PA, USA}
\acmBooktitle{Proceedings of the 11th ACM SIGPLAN International Conference on Certified Programs and Proofs (CPP '22), January 17--18, 2022, Philadelphia, PA, USA}

\begin{document}

\title{Implementing a Category-Theoretic Framework for Typed Abstract Syntax}

\author{Benedikt Ahrens}
\email{B.P.Ahrens@tudelft.nl}
\orcid{0000-0002-6786-4538}
\affiliation{%
  \institution{Delft University of Technology}
  \country{Netherlands}
}
\affiliation{%
  \institution{University of Birmingham}
  \country{United Kingdom}
}

\author{Ralph Matthes}
\email{ralph.matthes@irit.fr}
\orcid{0000-0002-7299-2411}
\affiliation{%
  \institution{IRIT, Université de Toulouse, CNRS, Toulouse INP, UT3}
  \city{Toulouse}
  \country{France}
}

\author{Anders Mörtberg}
\email{anders.mortberg@math.su.se}
\orcid{0000-0001-9558-6080}
\affiliation{%
  \institution{Department of Mathematics, Stockholm University}
  \city{Stockholm}
  \country{Sweden}
}


\begin{abstract}
  In previous work (``From signatures to monads in UniMath''), we
  described a category-theoretic construction of abstract syntax
  from a signature, mechanized in the \UniMath library based on the
  \Coq proof assistant.

  In the present work, we describe what was necessary to generalize
  that work to account for simply-typed languages.
  First, some definitions had to be generalized to account for the natural appearance of non-endofunctors in the simply-typed case.
  As it turns out,
  in many cases our mechanized results carried over to the generalized definitions without any code change.
  Second, an existing mechanized library on \omegacocont functors had to be
  extended by constructions and theorems necessary for constructing multi-sorted syntax.
  Third, the theoretical framework for the semantical signatures had to be generalized from a monoidal
  to a bicategorical setting, again to account for non-endofunctors arising in the typed case. This uses actions of endofunctors on functors
  with given source, and the corresponding notion of strong functors between actions, all formalized in \UniMath using a recently developed library of bicategory theory.
  We explain what needed to be done to plug all of these ingredients together, modularly.

  The main result of our work is a general construction that, when fed with a signature for a simply-typed language, returns an implementation of that language together with suitable boilerplate code, in particular, a certified monadic substitution operation.
\end{abstract}

\keywords{typed abstract syntax, monad, signature, formalization, computer-checked proof.}

\hyphenation{endo-functor endo-functors}

\maketitle

\section{Introduction}\label{sec:introduction}
There are many frameworks for the specification and analysis of abstract syntax.
Frequently, these frameworks are first developed for untyped syntax (see, \eg,~\citet{DBLP:conf/lics/FiorePT99,DBLP:conf/lics/Hofmann99,DBLP:conf/lics/GabbayP99}), with
extensions to the typed case often omitted, or promised as future work (see, \eg,~\citet{DBLP:journals/pacmpl/AhrensHLM20}).
However, in practice, such extensions are not trivial (see, \eg, \citet{DBLP:conf/ppdp/Fiore02,DBLP:conf/ppdp/MiculanS03}).
In the present paper, we report on the extension of one such framework---by \citet{signatures_to_monads}---from untyped to simply-typed syntax.
That work uses a notion of signature for untyped languages in terms of \emph{strong functors}; using Mendler iteration \citep{DBLP:journals/apal/Mendler91}, it builds the syntax generated by any such signature, together with a monadic substitution operation. The work is mechanized in the proof assistant \Coq \citep{Coq}, using the \UniMath library of univalent mathematics~\citep{UniMath}.

In the untyped case, an abstract signature is given by an endofunctor $H : [\CC,\CC] \to [\CC,\CC]$, for a suitably structured category~$\CC$ (\eg, $\set$).
For instance, the functor underlying the signature of the untyped $\lambda$-calculus is
\[ H~F = \Id{}_{\CC} + (F \times F) + (F \cdot \option)\enspace. \]
We write $\cdot$ for functor composition in
applicative order and functor application with a space, so $(F \cdot \option)~X = F~(X + 1)$, which hence
corresponds to $F$ with an additional free variable that will be bound by the abstraction constructor. With the aim of
initial algebra semantics
the \omegacocontinuity of $H$ is established and
an initial algebra for $H$ is constructed using a classical theorem of \citet{Adamek74}. Such
an initial algebra is an endofunctor $\Lambda : [\CC,\CC]$ with a(n
iso)morphism $\alpha : H~\Lambda \to \Lambda$ encoding the
constructors. The fact that this is an initial algebra furthermore ensures that $\Lambda$ has a suitable induction principle, making it equivalent to the following, more familiar, inductive type:

\begin{coqdoccode}
\coqdocemptyline
\coqdocnoindent
\coqdockw{Inductive} \coqdef{UniMath.SubstitutionSystems.STLC alt.InductiveLC.LC}{LC}{\coqdocinductive{LC}} (\coqdef{UniMath.SubstitutionSystems.STLC alt.X:77}{X}{\coqdocbinder{X}} : \coqref{UniMath.Foundations.Sets.hSet}{\coqdocdefinition{Type}}) : \coqref{UniMath.Foundations.Preamble.UU}{\coqdocdefinition{Type}} :=\coqdoceol
\coqdocindent{1.00em}
\ensuremath{|} \coqdef{UniMath.SubstitutionSystems.STLC alt.InductiveLC.Var}{Var}{\coqdocconstructor{Var}} : \coqref{UniMath.SubstitutionSystems.STLC alt.X:77}{\coqdocvariable{X}} \coqref{UniMath.Foundations.Init.::type scope:x 'xE2x86x92' x}{\coqdocnotation{→}} \coqref{UniMath.SubstitutionSystems.STLC alt.LC:78}{\coqdocinductive{LC}} \coqref{UniMath.SubstitutionSystems.STLC alt.X:77}{\coqdocvariable{X}}\coqdoceol
\coqdocindent{1.00em}
\ensuremath{|} \coqdef{UniMath.SubstitutionSystems.STLC alt.InductiveLC.App}{App}{\coqdocconstructor{App}} : \coqref{UniMath.SubstitutionSystems.STLC alt.LC:78}{\coqdocinductive{LC}} \coqref{UniMath.SubstitutionSystems.STLC alt.X:77}{\coqdocvariable{X}} \coqref{UniMath.Foundations.Init.::type scope:x 'xE2x86x92' x}{\coqdocnotation{→}} \coqref{UniMath.SubstitutionSystems.STLC alt.LC:78}{\coqdocinductive{LC}} \coqref{UniMath.SubstitutionSystems.STLC alt.X:77}{\coqdocvariable{X}} \coqref{UniMath.Foundations.Init.::type scope:x 'xE2x86x92' x}{\coqdocnotation{→}} \coqref{UniMath.SubstitutionSystems.STLC alt.LC:78}{\coqdocinductive{LC}} \coqref{UniMath.SubstitutionSystems.STLC alt.X:77}{\coqdocvariable{X}}\coqdoceol
\coqdocindent{1.00em}
\ensuremath{|} \coqdef{UniMath.SubstitutionSystems.STLC alt.InductiveLC.Lam}{Lam}{\coqdocconstructor{Lam}} : \coqref{UniMath.SubstitutionSystems.STLC alt.LC:78}{\coqdocinductive{LC}} (\coqref{UniMath.SubstitutionSystems.STLC alt.X:77}{\coqdocvariable{X}} \coqref{UniMath.SubstitutionSystems.STLC alt.InductiveLC.:::x '+' '1'}{\coqdocnotation{+}} \coqref{UniMath.SubstitutionSystems.STLC alt.InductiveLC.:::x '+' '1'}{\coqdocnotation{1}}) \coqref{UniMath.Foundations.Init.::type scope:x 'xE2x86x92' x}{\coqdocnotation{→}} \coqref{UniMath.SubstitutionSystems.STLC alt.LC:78}{\coqdocinductive{LC}} \coqref{UniMath.SubstitutionSystems.STLC alt.X:77}{\coqdocvariable{X}}.\coqdoceol
\coqdocemptyline
\end{coqdoccode}

By equipping $H$ also with a \emph{strength}, specifying how
substitution should act ``homomorphically'' on the constructors, the
functor $\Lambda$ may be equipped with a substitution operator
following \citet{DBLP:journals/tcs/MatthesU04} (formalized in
\UniMath by \citet{DBLP:conf/types/AhrensM15}). This substitution
satisfies the laws of a monad---an observation originating in the work
of \citet{DBLP:conf/csl/AltenkirchR99,DBLP:journals/scp/BellegardeH94,DBLP:journals/jfp/BirdP99}. Thus the
framework of \citep{signatures_to_monads} not only constructs datatypes
as initial algebras, but also equips them with a certified---in the
sense of satisfying formally verified laws---substitution
operation. To ease the use of the framework, the strong functor $H$ is
constructed algorithmically from a simple notion of signature---resulting in a form of datatype-generic programming in \UniMath.

In the present work, we extend this framework by both generalizing to the
simply-typed case, as well as refining and analyzing many of the
involved notions. Starting from a simple notion of multi-sorted
signature (\cref{sec:mssigs}), we construct (\cref{sec:endofunctor}) a corresponding signature functor
$H : [\CC,\CC] \to [\CC,\CC]$, for some suitably structured
category~$\CC$ (\eg, $\set^\sort$ for a (discrete, small) category $\sort$ of
sorts). This construction
involves functors that are not endo, so the notion of strength had to
be generalized to $H : [\CC,\CC] \to [\CC, \D]$---indeed, it
can easily be generalized to $[\CC,\D'] \to [\CC, \D]$, for
categories $\D$ and $\D'$. This generalization is purely formal and
the old \UniMath proofs worked without major changes. We find this quite
remarkable and a good example of how the initial investment of
mechanizing mathematics pays off in the long run---it would have been quite time-consuming to establish
sufficient conditions on $\D$ and $\D'$ by hand, and then verify that the proofs still go
through. On the way to this more
general notion of abstract signature we also decompose the notion of signature with strength of
\citep{DBLP:journals/tcs/MatthesU04} to clarify what is needed for the
substitution operation (\cref{sec:hetsubst}) and for it
to satisfy the monad laws (\cref{sec:hss}).

In order to construct the abstract datatype specified by $H$ as an
initial algebra, as well as the substitution monad, we prove that $H$
is \omegacocont. Many of the results already in \UniMath apply, but
some were missing, in particular results about when post-composing
with a functor is \omegacocont
(\cref{sec:omegacocont}). Having done this, the results of~\citep{signatures_to_monads} apply and we obtain the substitution
operation as a monad (\cref{sec:examples}). Throughout the paper we
rely on the simply-typed $\lambda$-calculus (STLC) as a running example,
but we have also implemented more complex languages, including PCF
(\cref{ex:pcf}) and the pre-syntax of the calculus of constructions
(\cref{ex:coc}). In the case of the STLC the obtained signature
functor $H$ is a bit more complex than in the untyped case (see
\cref{ex:sigSTLC}), but the generated syntax is equivalent to the following inductive
type:
\begin{coqdoccode}
\coqdocemptyline
\coqdocnoindent
\coqdockw{Inductive} \coqdef{UniMath.SubstitutionSystems.STLC alt.InductiveSTLC.STLC}{STLC}{\coqdocinductive{STLC}} (\coqdef{UniMath.SubstitutionSystems.STLC alt.X:66}{X}{\coqdocbinder{X}} : \coqref{UniMath.SubstitutionSystems.STLC alt.InductiveSTLC.S}{\coqdocaxiom{S}} \coqref{UniMath.Foundations.Init.::type scope:x 'xE2x86x92' x}{\coqdocnotation{→}} \coqref{UniMath.Foundations.Sets.hSet}{\coqdocdefinition{Type}}) : \coqref{UniMath.SubstitutionSystems.STLC alt.InductiveSTLC.S}{\coqdocaxiom{S}} \coqref{UniMath.Foundations.Init.::type scope:x 'xE2x86x92' x}{\coqdocnotation{→}} \coqref{UniMath.Foundations.Preamble.UU}{\coqdocdefinition{Type}} :=\coqdoceol
\coqdocindent{1.00em}
\ensuremath{|} \coqdef{UniMath.SubstitutionSystems.STLC alt.InductiveSTLC.Var}{Var}{\coqdocconstructor{Var}} : \coqdockw{\ensuremath{\forall}} \coqdef{UniMath.SubstitutionSystems.STLC alt.t:69}{t}{\coqdocbinder{t}}, \coqref{UniMath.SubstitutionSystems.STLC alt.X:66}{\coqdocvariable{X}} \coqref{UniMath.SubstitutionSystems.STLC alt.t:69}{\coqdocvariable{t}} \coqref{UniMath.Foundations.Init.::type scope:x 'xE2x86x92' x}{\coqdocnotation{→}} \coqref{UniMath.SubstitutionSystems.STLC alt.STLC:67}{\coqdocinductive{STLC}} \coqref{UniMath.SubstitutionSystems.STLC alt.X:66}{\coqdocvariable{X}} \coqref{UniMath.SubstitutionSystems.STLC alt.t:69}{\coqdocvariable{t}}\coqdoceol
\coqdocindent{1.00em}
\ensuremath{|} \coqdef{UniMath.SubstitutionSystems.STLC alt.InductiveSTLC.App}{App}{\coqdocconstructor{App}} : \coqdockw{\ensuremath{\forall}} \coqdef{UniMath.SubstitutionSystems.STLC alt.s:70}{s}{\coqdocbinder{s}} \coqdef{UniMath.SubstitutionSystems.STLC alt.t:71}{t}{\coqdocbinder{t}}, \coqref{UniMath.SubstitutionSystems.STLC alt.STLC:67}{\coqdocinductive{STLC}} \coqref{UniMath.SubstitutionSystems.STLC alt.X:66}{\coqdocvariable{X}} (\coqref{UniMath.SubstitutionSystems.STLC alt.s:70}{\coqdocvariable{s}} \coqref{UniMath.SubstitutionSystems.STLC alt.InductiveSTLC.:::x 'xE2x87x92' x}{\coqdocnotation{⇒}} \coqref{UniMath.SubstitutionSystems.STLC alt.t:71}{\coqdocvariable{t}}) \coqref{UniMath.Foundations.Init.::type scope:x 'xE2x86x92' x}{\coqdocnotation{→}} \coqref{UniMath.SubstitutionSystems.STLC alt.STLC:67}{\coqdocinductive{STLC}} \coqref{UniMath.SubstitutionSystems.STLC alt.X:66}{\coqdocvariable{X}} \coqref{UniMath.SubstitutionSystems.STLC alt.s:70}{\coqdocvariable{s}} \coqref{UniMath.Foundations.Init.::type scope:x 'xE2x86x92' x}{\coqdocnotation{→}} \coqref{UniMath.SubstitutionSystems.STLC alt.STLC:67}{\coqdocinductive{STLC}} \coqref{UniMath.SubstitutionSystems.STLC alt.X:66}{\coqdocvariable{X}} \coqref{UniMath.SubstitutionSystems.STLC alt.t:71}{\coqdocvariable{t}}\coqdoceol
\coqdocindent{1.00em}
\ensuremath{|} \coqdef{UniMath.SubstitutionSystems.STLC alt.InductiveSTLC.Lam}{Lam}{\coqdocconstructor{Lam}} : \coqdockw{\ensuremath{\forall}} \coqdef{UniMath.SubstitutionSystems.STLC alt.s:72}{s}{\coqdocbinder{s}} \coqdef{UniMath.SubstitutionSystems.STLC alt.t:73}{t}{\coqdocbinder{t}}, \coqref{UniMath.SubstitutionSystems.STLC alt.STLC:67}{\coqdocinductive{STLC}} (\coqref{UniMath.SubstitutionSystems.STLC alt.X:66}{\coqdocvariable{X}} \coqref{UniMath.SubstitutionSystems.STLC alt.InductiveSTLC.:::x '+' 'x7B' x 'x7D'}{\coqdocnotation{+}} \coqref{UniMath.SubstitutionSystems.STLC alt.InductiveSTLC.:::x '+' 'x7B' x 'x7D'}{\coqdocnotation{\{}} \coqref{UniMath.SubstitutionSystems.STLC alt.s:72}{\coqdocvariable{s}} \coqref{UniMath.SubstitutionSystems.STLC alt.InductiveSTLC.:::x '+' 'x7B' x 'x7D'}{\coqdocnotation{\}}}) \coqref{UniMath.SubstitutionSystems.STLC alt.t:73}{\coqdocvariable{t}} \coqref{UniMath.Foundations.Init.::type scope:x 'xE2x86x92' x}{\coqdocnotation{→}} \coqref{UniMath.SubstitutionSystems.STLC alt.STLC:67}{\coqdocinductive{STLC}} \coqref{UniMath.SubstitutionSystems.STLC alt.X:66}{\coqdocvariable{X}} (\coqref{UniMath.SubstitutionSystems.STLC alt.s:72}{\coqdocvariable{s}} \coqref{UniMath.SubstitutionSystems.STLC alt.InductiveSTLC.:::x 'xE2x87x92' x}{\coqdocnotation{⇒}} \coqref{UniMath.SubstitutionSystems.STLC alt.t:73}{\coqdocvariable{t}}).\coqdoceol
\coqdocemptyline
\coqdocnoindent
\end{coqdoccode}
Here, \coqdocaxiom{S} is a small set of sorts closed under \coqdocnotation{⇒},
and \coqdocvariable{X} is a set of sorted variables indexed by
\coqdocaxiom{S}.  Furthermore,
\coqref{UniMath.SubstitutionSystems.STLC alt.X:66}{\coqdocvariable{X}} \coqref{UniMath.SubstitutionSystems.STLC alt.InductiveSTLC.:::x '+' 'x7B' x 'x7D'}{\coqdocnotation{+}} \coqref{UniMath.SubstitutionSystems.STLC alt.InductiveSTLC.:::x '+' 'x7B' x 'x7D'}{\coqdocnotation{\{}} \coqref{UniMath.SubstitutionSystems.STLC alt.s:72}{\coqdocvariable{s}} \coqref{UniMath.SubstitutionSystems.STLC alt.InductiveSTLC.:::x '+' 'x7B' x 'x7D'}{\coqdocnotation{\}}}
denotes \coqdocvariable{X} with an additional variable of
sort \coqdocvariable{s}.

Having successfully generalized the formalization in \citep{signatures_to_monads} to simply-typed signatures, we
give in \cref{sec:understanding} a bicategorical analysis
of the notion of strength of \cref{def:strength}. To this end we
relate it to an action-based notion of strength as put forward by
\citet{Pareigis1977}, which also occurs in the work of
\citet{DBLP:conf/lics/Fiore08}. This is then also related to a notion
of relative strength of \citep{DBLP:conf/types/AhrensM15}.  This
section makes use of displayed categories
\citep{DBLP:journals/lmcs/AhrensL19} and of a library of bicategory
theory \citep{DBLP:conf/rta/AhrensFMW19} recently added to \UniMath.

\subsection{Contributions}
We present a fully formalized mathematical framework for simply-typed
syntax relying on a categorical construction of inductive families as
initial algebras. In this way, we avoid meta-programming and complex
nested inductive types, and instead work completely internally to
\UniMath. We generically equip the constructed syntax with a
monadic substitution operation. Furthermore, we provide a
bicategorical analysis of the involved notions, motivated by
the generalization of strength needed to encompass multi-sorted
languages.

All results in this paper have been formalized in \UniMath. The code is integrated into the \UniMath library \citep{UniMath}.
For documentation, we refer to a specific version of \UniMath, available
in commit \href{https://github.com/UniMath/UniMath/commit/\longhash}{\longhash}.
The HTML documentation derived from this version is hosted \href{\coqdocbasebaseurl/toc.html}{\color{blue}{online}}.
To connect the paper and the formalization we provide
clickable identifiers (\eg, \coqident{SubstitutionSystems.MultiSorted_alt}{MultiSortedSig}) for all definitions and results. By
clicking on one of the links, one gets taken to the corresponding
definition in the HTML documentation.

\subsection{Related Work}
\label{sec:related-work}

There are many theoretical frameworks suitable for the study of abstract syntax, see, e.g., the work of \citet{DBLP:conf/ppdp/MiculanS03,Tanaka:2005:UCF:1088454.1088457,DBLP:conf/lics/Fiore08,DBLP:conf/ppdp/Fiore02,DBLP:conf/types/GambinoH03,DBLP:conf/fscd/HirschowitzHL20,DBLP:conf/lics/FioreH13,DBLP:journals/jfp/AltenkirchGHMM15,DBLP:conf/icfp/ChapmanDMM10,DBLP:conf/fossacs/Hamana11,DBLP:conf/icalp/Fiore12,DBLP:conf/icfp/HamanaF11}.
We are interested in the following properties of such a framework:
\begin{enumerate}
\item Is there a computer-checked implementation? \label{item:computer-checked}
\item Is there an internal notion of signature? \label{item:internal-sig}
\item Does it support multi-sorted languages?
\item What is the theory that it is implemented in?
\end{enumerate}

\subsubsection{Implementations}
In this section, we only consider frameworks that satisfy property~\ref{item:computer-checked}.
We classify them into three categories according to property~\ref{item:internal-sig}.

Into the first category fall frameworks that work with an \emph{external} notion of signature.
Such frameworks take a signature specified in some domain-specific language, \eg, as a \Haskell data type, and return a library of computer code for some proof checker such as \Coq.%
\footnote{This approach is called ``generative'' by~\citet{DBLP:conf/esop/LeeOCY12}.}
To this category belong the tools \systemname{DBGen}~\citep{DBLP:conf/itp/Polonowski13}, \systemname{Ott}~\citep{DBLP:journals/jfp/SewellNOPRSS10}, \systemname{LNGen}~\citep{lngen}, \systemname{Autosubst}~\citep{DBLP:conf/itp/SchaferTS15} and \systemname{Autosubst2}~\citep{DBLP:conf/cpp/StarkSK19}. These are different from our approach in that they typically rely on some form of meta-programming as the notion of signature is external.

The second category consists of work building a library of generic functions and reasoning principles for (typed) abstract syntax.
The library \systemname{NomPa}~\citep{nompa, pouillard:tel-00759059} falls into this category, as does \systemname{GMETA}~\cite{DBLP:conf/esop/LeeOCY12}.
In this line of work, there are several modules called signatures, providing an interface for signatures. However, there is not a fixed mathematical definition of a signature, making them less suited for general mathematical study of syntax and more focused on convenience for specifying and working with a specific language.

Into the third category fall tools that have an \emph{internal} definition of signatures, stated within the proof assistant.
In this category, signatures are themselves first-class objects and may be reasoned about within the system.
Our work falls into this category, and we thus spend some more time comparing it to other work in this category.

Ahrens~\citep{DBLP:journals/mscs/Ahrens16} considers only untyped syntax. The work is fully implemented in the computer proof assistant \Coq, and relies on \Coq's inductive types and associated structural recursion for the construction of syntax and substitution.

Work by \citet{DBLP:conf/rta/AhrensHLM19,DBLP:journals/lmcs/AhrensHLM21} only discusses untyped syntax, but construct, like us, syntactic models without relying on inductive types.
\citet{DBLP:journals/jfrea/AhrensZ11} consider a simple notion of signature akin to our multi-sorted signatures of \cref{def:mssigs}.
To any such signature, they associate a category of models and construct an initial such model.
For this construction, crucially, they rely on suitable mutually inductive type families.
Compared to their work, we study here a more general notion of abstract signatures (\cref{def:strength}) together with a map from multi-sorted signatures to these abstract signatures.
We then construct the initial models without relying on general inductive types, but rather on a category-theoretic construction of initial algebras.
A set-theoretic construction of initial models closer to ours is given in the PhD thesis of \citet{ju_phd}, but is not computer-checked.

\citet{DBLP:journals/pacmpl/AllaisA0MM18,DBLP:journals/corr/abs-2001-11001} consider ``descriptions'' (see also~\citep{DBLP:conf/lics/DagandM13}) for signatures.
Like \citet{DBLP:journals/jfrea/AhrensZ11}, they assume suitable inductive datatypes; descriptions then straightforwardly generate relative ``term'' monads. These monads should be initial in a suitable category, even though such categorical formulations are not used or explored in that work.
In both works, a well-behaved monadic substitution operation is constructed via the induction mechanism of \Agda.
The framework of \citet{DBLP:conf/icfp/LohM11} is suitable for specifying abstract syntax with binding; typed syntax is not discussed. Again, the construction of syntax relies on \Agda's datatypes.

In \cite{fiore-szamozvancev}, the authors study two notions of signature similar to the ones studied here---multi-sorted binding signatures and signatures with strength---and also give a translation from the former to the latter, similar to our translation in \cref{sec:msbstoss}. Like \cite{DBLP:journals/jfrea/AhrensZ11} they use inductive datatypes to construct syntax.
Their notion of syntax includes ``meta-variables''.

In summary, compared to the aforementioned work, we study here the construction of syntax,
and of a suitable substitution operation for syntax, from a \emph{category-theoretic} rather than a \emph{type-theoretic} perspective.
Specifically, we construct syntax exploiting \omegacocontinuity of the signature functor,
and substitution for that syntax using Mendler-style recursion, expressed categorically.

\subsubsection{Comparison to Other Frameworks}
Here, we review categorical frameworks for abstract syntax that are not, to our knowledge, mechanized.

Fiore, Plotkin, and Turi~\cite{DBLP:conf/lics/FiorePT99} consider untyped syntax specified by an untyped variant of binding signatures (\cref{def:mssigs}).
The authors consider both a categorical formulation of recursion to specify substitution (on ``free'' objects),
as well as a type-theoretic formulation using inductive types and structural recursion.
However, no link is made between the category-theoretic and the type-theoretic formulations.
The mathematical structure given by substitution and its properties is taken to be a ``substitution monoid''.

Fiore~\cite{DBLP:conf/lics/Fiore08} considers a notion of signature (given by an endofunctor $\Sigma : \CC \to \CC$) with strength, like we do in the present work. They discuss the categorical construction of syntax and substitution as the colimit of an $\omega$-cocontinuous functor, see \cite[Cor.~4]{DBLP:conf/lics/Fiore08}.
We discuss their notion of signature in more detail in~\cref{sec:understanding}.
Fiore and Hamana~\cite{DBLP:conf/lics/FioreH13} later consider polymorphically typed terms, such as System F.
They do not discuss the construction of abstract syntax.
Hur~\cite{DBLP:phd/ethos/Hur10} considers, like we do here, category-theoretic constructions of syntax.
The focus of that work is on equations between terms.
Mahmoud~\cite{DBLP:phd/ethos/Mahmoud11} considers structural induction in a type-theoretic style;
a category-theoretic construction of syntax or substitution is not discussed.

Work by Hirschowitz and Maggesi~\cite{DBLP:conf/wollic/HirschowitzM07, DBLP:journals/iandc/HirschowitzM10} considers syntax and substitution in form of monads on the category of sets.
Hirschowitz and Maggesi also implement a special case of their work---for the $\lambda$-calculus---in the computer proof assistant \Coq.
(A comparison between Fiore, Plotkin, and Turi's approach~\citep{DBLP:conf/lics/FiorePT99} using substitution monoids and Hirschowitz and Maggesi's using monads is given in Zsidó's thesis~\cite{ju_phd}.)
The first part of the present work, in particular, of \cref{sec:msbstocs}, can be considered to be ``zooming in'' on Hirschowitz and Maggesi's work, notably on the construction of the initial monad, expressed purely categorically.

Ahrens~\citep{DBLP:conf/wollic/Ahrens12,DBLP:journals/corr/abs-1107-4751,DBLP:journals/lmcs/Ahrens19} considers simply-typed signatures and the construction of the generated syntax in a type-theoretic setting.
Some instances of the framework are formalized in \Coq, but general signatures are not.

Hirschowitz, Hirschowitz, and Lafont~\cite{DBLP:conf/fscd/HirschowitzHL20} consider a very abstract notion of signature.
They construct abstract syntax via categorical machinery, and furthermore the passage to a ``quotient syntax'' modulo a system of equations.
In the extended version~\cite[Sec.~5.5]{HHL20-ext}, the authors explain the construction of a signature in their sense from a signature à la Fiore, Plotkin, and Turi~\cite{DBLP:conf/lics/FiorePT99}.

\subsection{UniMath and Category Theory Therein}

In this section we provide a brief introduction to \UniMath and fix notations used throughout the paper.
We write $a = b$ for the type of identifications/equalities/paths from $a$ to $b$, and $a \convert b$ for definitional/judgmental equality---in particular, we write $a \eqdef b$ for defining $a$ to be $b$.
We do not rely on any inductive types other than the ones specified in the prelude of \UniMath, such as identity types, sum types, natural numbers, and booleans.
Indeed, part of the work consists of the construction of certain initial algebras from dependent products, dependent sums, identities, and natural numbers.
We use the notions of propositions and sets of Univalent Foundations:
a type $X$ is a proposition if $\prd{x,y:X}x=y$ is inhabited,
and a set if the type $x = y$ is a proposition for all $x, y : X$.
Hence, despite working in \Coq, we do not rely on the universes \coqdocdefinition{Prop} or \coqdocdefinition{Set}.

However, our main constructions and results,
while conveniently expressed using univalent foundations,
are not dependent on the full univalence axiom.
On top of Martin-Löf type theory~\citep{MLTT}, they rely on
propositional and functional extensionality, and propositional
resizing.
The key consequence of these that we use is the construction of effective set quotients due to \citet{DBLP:journals/mscs/Voevodsky15}.
These axioms are compatible with the Uniqueness of Identity Proofs (UIP) principle,
and hence with the interpretation of types as sets.

We assume there to be at least two type-theoretic universes, $\UU_0 : \UU_1$.
We call a type $X : \UU_0$ a \emph{small} type.
Following the convention established in the HoTT book~\cite{hottbook}, we leave the universe levels implicit throughout the paper.

A \emph{category} $\CC$ in \UniMath is given by a type of objects $\ob{\CC}$ and a family of sets $\CC(x,y)$ for any $x, y : \ob{\CC}$.
We call $\CC$ small if objects and morphisms are types from $\UU_0$.
We assume the reader to be familiar with the concepts of category
theory as found in the standard text by \citet{maclane}.
We sometimes remark that specific categories are univalent \citep{rezk_completion}. Such results typically depend on the full univalence axiom, but none of our main results depend on categories being univalent. The reader who prefers to read our results outside of univalent foundations may hence safely ignore these remarks.

The category $\set \eqdef \set(\UU_0)$ of sets has, as objects, sets from the universe $\UU_0$; this is a category whose types live in $\UU_1$, hence it is not small.
The set quotients discussed above are used to ensure that $\set$ is cocomplete (\ie that it has small colimits).
Note that the correct use of universe levels is not checked by \UniMath---we hence have to keep track of them ourselves.
This is particularly important when taking (co)limits in $\set$;
we ensure that we only take such (co)limits indexed by a small graph, \ie, a graph whose types of nodes and vertices are elements of $\UU_0$.

\section{From Multi-Sorted Binding Signatures to Signatures with Strength}\label{sec:msbstoss}
In this section we give a category-theoretic notion of abstract
signature in the form of strong functors. But first we will consider
our main source of examples: signature functors arising from a notion
of multi-sorted binding signatures.

\subsection{Multi-Sorted Binding Signatures}\label{sec:mssigs}

In what follows, we fix a type $\sort$ representing the sorts.
Many of our constructions will rely on $\sort$ being a set;
in the formalization, this assumption is explicitly stated when needed,
but in this paper we take $\sort$ to be a set everywhere, for sake of simplicity.

We start with the following
definition of multi-sorted binding signatures:

\begin{definition}[\coqident{SubstitutionSystems.MultiSorted_alt}{MultiSortedSig}]\label{def:mssigs}
  A \emph{multi-sorted binding signature} is given by a small set $I$
  together with an arity function
  $\arity : I \to \lst (\lst(\sort) \times \sort) \times \sort$.
\end{definition}

Intuitively, for any $i : I$, there is a term constructor described by
its ``arity'' $\arity(i)$, whose first component describes the list of
arguments: for each argument we specify a list of sorts for the
variables bound in it, as well as its sort. The second component of
$\arity(i)$ designates the sort of the term constructed by the
constructor pertaining to the index~$i$. This definition of
multi-sorted signatures is not new, and variations on it can be found
in the literature, e.g. in the work of
\citet[Definition~4.3]{DBLP:journals/jfrea/AhrensZ11}. The standard example is that of
the simply-typed $\lambda$-calculus (STLC).

\begin{example}[\coqident{SubstitutionSystems.STLC_alt}{STLC_Sig}]\label{example:stlc}
  Assume that $\sort$ is closed under a binary operation
  $\Rightarrow\,: \sort \to \sort \to \sort$ representing function
  types. We have to put into $I$ the sort parameters of the typing
  rules of the term constructors of STLC. Thus, $I$ is taken to be
  $(\sort \times \sort) + (\sort \times \sort)$. The left summand
  pertains to the application operation while the right summand
  describes $\lambda$-abstraction:
  \begin{align*}
    \arity(\tinl\langle s,t\rangle)&\eqdef \big\langle[\langle[],s\Rightarrow t\rangle,\langle[],s\rangle],t\big\rangle\enspace,\\
    \arity(\tinr\langle s,t\rangle)&\eqdef \big\langle[\langle[s],t\rangle],s\Rightarrow t\big\rangle\enspace.
  \end{align*}

  For sorts $s$ and $t$, there is an application constructor
  taking as input two terms, of sorts $s\Rightarrow t$ and $s$, respectively, to
  yield a term of sort $t$. Again, for sorts $s$ and $t$, there is
  an abstraction constructor taking as input a term of sort $t$
  with a an additional free variable of sort $s$, and yielding a
  term of sort $s \Rightarrow t$. This is close to the
  usual presentation of the typing rules of STLC. However, we do
  not include a constructor for variables as these
  will be added generically in \cref{sec:hetsubst}.
  Note that this is a form of locally nameless presentation: bound variables have canonical names.
\end{example}

\subsection{Functors from Multi-Sorted Binding Signatures}\label{sec:endofunctor}

Recall that $\ccsort$ is the functor category $[\sort,\CC]$ where
$\sort$ is viewed as a discrete category. In the case when $\CC$ is
$\set$, objects of this categories are simply functions
$X : \sort \to \set$. These hence correspond to sets of sorted
variables analogously to the parameter \coqdocvariable{X} of the
example of the inductive family \coqdocinductive{STLC} in \cref{sec:introduction}.
Equivalently, we could consider families of objects in $\CC$ indexed over
$\sort$, but $[\sort,\CC]$ has the advantage that results on functor
categories (constructions of (co)limits, etc) apply directly.

To construct $H$, we assume that $\sort$ is a set and that $\CC$ has a
terminal object $1$, binary products, and set-indexed coproducts.
We first define a few helper functors:

\begin{definition}[\coqident{SubstitutionSystems.MultiSorted_alt}{sorted_option_functor}]
  Let $s$ be a sort. The \emph{sorted option functor}
  $\option_s : \ccsort \to \ccsort$ is defined as
  \[
    \option_s~X~t \eqdef X~t \, + \, \coprod_{(s = t)} 1\enspace.
  \]
\end{definition}

\begin{remark}
  The rationale behind the definition is that when given $X : \ccsort$
  and $t : \sort$ we get that $\option_s~X~t$ is $X$ with an extra
  variable of sort $s$ if $s$ and $t$ are equal. If
  $\sort$ has decidable equality, $\eqdec$, this can be
  defined as
  \[
    \option_s~X~t \eqdef \ifte{(s \eqdec t)}{(X~t + 1)}{(X~t)}\enspace.
  \]

  However, in order to avoid assuming decidable equality for $\sort$
  we do the above trick where we form a coproduct of $1$ over the type
  of proofs that $s = t$ (\ie, we form a subsingleton). For this
  coproduct to exist, \eg, in the category of sets, it would
  suffice to assume that $\sort$ is a (small) $1$-type so that
  $s = t$ is a (small) set. However, this would make $\option_s~X~t$
  add as many variables as there are proofs of $s = t$, which is not
  what we intend with the definition.
\end{remark}

\begin{definition}[\coqident{SubstitutionSystems.MultiSorted_alt}{option_list}]
  Given a non-empty list of sorts $\ell \convert [s_1, \ldots, s_n]$,
  we define $\option~\ell : \ccsort \to \ccsort$ as
  \[
    \option~\ell \eqdef \option_{s_1} \hcomp \, (\option_{s_2} \hcomp \, \ldots)\enspace.
  \]
  For an empty list, we define $\option~[] \eqdef \Id{}$.
\end{definition}

\begin{definition}[\coqident{SubstitutionSystems.MultiSorted_alt}{projSortToC}]
  For any $s : \sort$ we have a \emph{projection functor}
  $\proj{s} : \ccsort \to \CC$ defined as:
  \[
    \proj{s}~X \eqdef X~s\enspace.
  \]
\end{definition}

\begin{definition}[\coqident{SubstitutionSystems.MultiSorted_alt}{hat_functor}]
  For any $s : \sort$ we have a left adjoint\footnote{The fact that
    these functors are adjoint is proved in \cref{lem:hat}.} to
  $\proj{s}$, written $\hat{s} : \CC \to\ccsort$, defined as
  \[
    \hat{s}~X~t \eqdef \coprod_{(s = t)} X\enspace.
  \]
\end{definition}

\begin{remark}
  Once again we use the ``coproduct-trick'' to avoid assuming
  decidable equality on $\sort$. If we had decidable equality then
  $\hat{s}$ could be defined as
  \[
    \hat{s}~X~t \eqdef \ifte{(s \eqdec t)}{X}{0}\enspace,
  \]
  where $0$ is the initial object in $\CC$ (which exists since $\CC$ has
  set-indexed coproducts).
\end{remark}

We can now define the functor $H$ in multiple steps. Given
$a = (\ell,s) : \lst(\sort) \times \sort$ specifying the sorts of
bound variables and type of an argument to some constructor we first
define the object part of a functor
$F^{a} : [\ccsort,\ccsort] \to [\ccsort,\CC]$ as
\[
  F^a~X \eqdef \proj{s} \hcomp X \hcomp \option~\ell\enspace.
\]

That is, $F^a$ itself is the composition of ``precomposition with
$\option~\ell$'' and ``postcomposition with $\proj{s}$''. Note that
$F^a$ is \emph{not} an endofunctor on a category of endofunctors. This
will make it necessary, in the next section, to give a general notion
of signature with strength to be able to analyze $F^a$.

Given an arity $(\vec{a},t)$ with $t : \sort$ and
$\vec{a} = [a_1, \ldots, a_m]$ with each $a_i = (\ell,s)$ as above, we
define the object part of an endofunctor $F^{(\vec{a},t)}$ on
$[\ccsort,\ccsort]$ as
\[
  F^{(\vec{a},t)}~X \eqdef \hat{t} \hcomp (F^{a_1}~X \times \ldots \times F^{a_m}~X)\enspace.
\]

More formally, $F^{(\vec{a},t)}$ is obtained by first forming the
pointwise product of $F^{a_1}, \ldots, F^{a_m}$ and then by composing
with ``postcomposition with $\hat{t}$''.  In the case when $\vec{a}$
is empty, we compose $\hat{t} \hcomp \_$ with the functor that
constantly outputs $1$.

Combining all of this we get the definition of $H$:

\begin{definition}[\coqident{SubstitutionSystems.MultiSorted_alt}{MultiSortedSigToFunctor}]\label{def:functor_from_binding_sig}
  Given a multi-sorted binding signature $(I,\arity)$, its associated
  signature functor $H : [\ccsort,\ccsort] \to [\ccsort,\ccsort]$ is
  given by the following (pointwise) coproduct:
  \[
    H~X \eqdef \coprod_{i~:~I} F^{\,\arity(i)}(X)\enspace.
  \]
\end{definition}

Note that the above coproduct exists as $I$ is assumed to be a set in
\cref{def:mssigs}.

\begin{example}[\coqident{SubstitutionSystems.STLC_alt}{app_source}, \coqident{SubstitutionSystems.STLC_alt}{lam_source}]
  \label{ex:sigSTLC}
  We apply the general construction and obtain a functor which
  is the coproduct of two families of functors, $\app$ and
  $\lam$, indexed by $s, t : \sort$. These functors are
  defined pointwise at $X : [\ccsort,\ccsort]$ as:
  \begin{align*}
    \app~X &\eqdef \hat{t} \hcomp ((\proj{s \Rightarrow t} \hcomp X) \times (\proj{s} \hcomp X))\enspace, \\
    \lam~X &\eqdef \widehat{s \Rightarrow t} \hcomp (\proj{t} \hcomp X \hcomp \option_s)\enspace.
  \end{align*}

  To obtain exactly these functors in the formalization,
  special care is needed in $F^{(\vec{a},t)}$ to avoid taking the
  product with the constant functor. These details make it a little
  less direct to formalize the strength laws and the proof of
  \omegacocontinuity; we refer the interested reader to the
  formalization.
\end{example}

\begin{remark}
  There are operations for taking the disjoint sum of signatures, on
  both multi-sorted binding signatures and signatures with
  strength. They can be used to assemble complicated signatures from
  simpler ones, as we do in the example of PCF in~\cref{ex:pcf}. The
  translation of signatures presented here preserves these sums.
\end{remark}

\subsection{Extending the Signature Functor for Substitution}\label{sec:hetsubst}

Since variables play a role that is different from all other term
constructors when it comes to substitution, the endofunctor $H$ is not
supposed to comprise the inclusion of variables into the terms. The
``variable case'' is added explicitly to $H$, by considering the
endofunctor $\constfunctor{\Id{\ccsort}}+H$ on $[\ccsort,\ccsort]$. Here $\constfunctor{\Id{\ccsort}}$ is the functor that is constantly $\Id{}$. So on objects, $\constfunctor{\Id{\ccsort}}+H$
associates with $X:[\ccsort,\ccsort]$ the endofunctor $\Id{\ccsort}+H~X$ on
$\ccsort$. Accordingly, an $(\constfunctor{\Id{\ccsort}}+H)$-algebra consists of
$T:[\ccsort,\ccsort]$ and $\alpha:\Id{\ccsort}+H~T\to T$. It is convenient to
present $\alpha$ (uniquely) as $[\eta,\tau]$ with $\eta:\Id{\ccsort}\to T$ and
$\tau:H~T\to T$.%
\footnote{Note that we are not assuming that $(T,\alpha)$ is an \emph{initial} $(\constfunctor{\Id{\ccsort}}+H)$-algebra.
The construction of initial $(\constfunctor{\Id{\ccsort}}+H)$-algebras is considered in Section~\ref{sec:msbstocs}.}

We will express substitution in the style of a monad in monoid form,
thus with the aforementioned $T$ and
$\eta$, and moreover a natural transformation
$\prejoin:T\hcomp T\to T$ that is typically called $\mu$ in the
literature. We use the triple format since we want to
base our implementation on that of \citep{signatures_to_monads} which
in turn relies on the development of
\citep{DBLP:journals/tcs/MatthesU04} that uses monoid style monads.

In the following definition, we do not care about all the monad laws, but about the
specification of the behavior of $\prejoin$ (in other words, the ``functional properties'' that characterize it).
Crucially, we address
the desideratum that substitution is ``homomorphic'' on the
domain-specific constructors---as specified by the multi-sorted
binding signature.
The next definition is in the spirit of
\citep{DBLP:journals/tcs/MatthesU04},
but isolates more clearly what is
needed to specify the properties of substitution itself.

\begin{definition}[\coqident{SubstitutionSystems.SubstitutionSystems}{heterogeneous_substitution}]\label{def:hetsubst}
  Given $H$ as specified above, a \fat{heterogeneous substitution} consists
  of an $(\constfunctor{\Id{\ccsort}}+H)$-algebra
  $(T,\eta,\tau)$ and a natural transformation
  \[\th{}:(H {-}) \hcomp T \naturalto H ({-} \hcomp T) : [\ccsort,\ccsort] \to [\ccsort,\ccsort]\] between functors of the same type as $H$, such
  that there is a unique $[\ccsort,\ccsort]$-morphism $\prejoin:T\hcomp T\naturalto T$
  making the diagram in \cref{fig:hetsubst} commute.
\begin{figure}[tb]
\[
\begin{array}{c}
\xymatrix@C3pc@R1pc{
T \ar[r]^{\et{} \hcomp T} \ar[ddr]_{1_{T}}
    & T \hcomp T \ar[dd]^{\prejoin}
        & (H\, T) \hcomp T  \ar[l]_-{\ta{} \hcomp T} \ar[d]^{\th{T}} \\
    &
        & H (T \hcomp T) \ar[d]^{H\, \prejoin} \\
    & T
        & H\, T  \ar[l]_-{\ta{}}
}
\end{array}
\]
\vspace{-0.4cm}
\caption{\coqident{SubstitutionSystems.SubstitutionSystems}{bracket_property_parts_identity_nicer}}\label{fig:hetsubst}
\vspace{-0.6cm}
\Description{}
\end{figure}
\end{definition}

Notice that the functor $T$ in the upper left corner of
\cref{fig:hetsubst} is equal to $\Id{}\hcomp T$, in a way which makes
the type of arrows out of both convertible, and this allows to
type-check $\eta \hcomp T$.

  We recognize, in the triangle part of the diagram,
  one of the monad laws---the law
 saying, intuitively, that substitution for variables is done by
look-up. The rectangle part is the specification of ``homomorphic''
behavior of $\prejoin$ on the other term constructors, as instructed by
$\theta$. The natural transformation $\theta$ is only used with parameter $T$, but enters fully
into the later theorems on unique existence of $\prejoin$ when
$(T,\eta,\tau)$ is an initial algebra. The notion is not limited to
initial algebras (\citep{DBLP:journals/tcs/MatthesU04} develops the
theory for non-wellfounded syntax as well), but we will only use it
for those. We do not claim that this data suffices to obtain the other
monad laws that are shown in
\cref{fig:othermonadlaws}. Note that the functor $T$ in the upper
left corner is equal to $T\hcomp\Id{}$, again in a way which allows to type-check $T \hcomp \eta$.
Likewise, we use associativity of functor
composition in the upper right corner that holds in a way that allows to type-check the two arrows out
of it, for both possible associations of the expression. If the monad laws are satisfied, the usual name for $\prejoin$ is ``$\join$''.
\begin{figure}[tb]
\[
\begin{array}{c}
\xymatrix@C3pc@R2pc{
T \ar[r]^{T \hcomp \et{}} \ar[dr]_{1_{T}}
    & T \hcomp T \ar[d]^{\prejoin}
        & T\hcomp T \hcomp T  \ar[l]_-{T \hcomp \prejoin} \ar[d]^{\prejoin\hcomp T} \\
    & T
        & T\hcomp T  \ar[l]_-{\prejoin}
}
\end{array}
\]
\vspace{-0.3cm}
\caption{Second and third part of \coqident{CategoryTheory.Monads.Monads}{Monad_laws_pointfree}}\label{fig:othermonadlaws}\vspace{-0.3cm}
\Description{}
\end{figure}

\subsection{Adding Strength to Ensure the Monad Laws}\label{sec:hss}

As mentioned above, the unique existence of $\prejoin$ according to
\cref{def:hetsubst} will not suffice to guarantee that
$(T,\eta,\prejoin)$ becomes a monad.
In this section, we take three simple steps to ensure that our signature functor $H$ does generate a monad.
The second and third step are intertwined, and are discussed together.

\subsubsection{Generalizing the Type of $\theta$}
Firstly, more natural
transformations have to be uniquely determined by a generalization of
\cref{fig:hetsubst} which depends on a version of $\theta$ with two
arguments. Given a signature functor $H$ as before, we will call
\emph{prestrength}
for $H$ any natural transformation
$\th{}:(H {-}) \hcomp U {\sim} \arr H ({-} \hcomp U {\sim})$ between
bifunctors
$[\ccsort,\ccsort] \times \Ptd(\ccsort) \arr [\ccsort,\ccsort]$, with
the category $\Ptd(\ccsort)$ of pointed endofunctors
over $\ccsort$
and its forgetful functor $U : \Ptd(\ccsort) \to [\ccsort,\ccsort]$. This  notion of prestrength will be an instance of the
less constrained \cref{def:prestrength} below.
In the situation of
\cref{def:hetsubst}, we have the pointed endofunctor $(T,\eta)$, and
if we fix the second argument (symbolized by $\sim$) of a prestrength
$\th{}$ to $(T,\eta)$, we get (by virtue of $U(T,\eta) \convert T$) a natural
transformation of the type required in that definition.
\begin{definition}[\coqident{SubstitutionSystems.SubstitutionSystems}{hss}]\label{def:hss}
  Given $H$ as before and a prestrength $\theta$ for $H$,
 an $(\constfunctor{\Id{\CC}}+H)$-algebra $(T,\eta,\tau)$ is a
\fat{heterogeneous substitution system} (HSS) for
$(H, \th{})$, if, for every $\Ptd(\ccsort)$-morphism $f : (Z,e) \arr (T,
\et{})$, there exists a unique $[\ccsort,\ccsort]$-morphism $h : T \hcomp Z \arr
T$ making the diagram in \cref{fig:hss} commute. This morphism is denoted $\gst{(Z,e)}{f}$.
\begin{figure}[tb]
\[
\begin{array}{c}
\xymatrix@C3pc@R1pc{
Z \ar[r]^{\et{} \hcomp Z} \ar[ddr]_{Uf}
    & T \hcomp Z \ar[dd]^{h}
        & (H T) \hcomp Z  \ar[l]_-{\ta{} \hcomp Z} \ar[d]^{\th{T, (Z,e)}} \\
    &
        & H (T \hcomp Z) \ar[d]^{H h} \\
    & T
        & H T  \ar[l]_-{\ta{}}
}
\end{array}
\]
\vspace{-0.4cm}
\caption{\coqident{SubstitutionSystems.SubstitutionSystems}{bracket_property_parts}}\label{fig:hss}
\vspace{-0.4cm}
\Description{}
\end{figure}
\end{definition}
Given our previous remark on fixing the second argument of $\th{}$ to
$(T,\eta)$, the uniquely existing $\prejoin$ of a heterogeneous
substitution is the special case with $f \eqdef \idwt{}_{(T,\eta)}$ (\ie, the identity natural transformation).
Analogously
to \cref{def:hetsubst}, the first argument of $\th{}$ is always set to
$T$ in this diagram, but needs to vary when proving that one obtains
an HSS when $(T,\eta,\tau)$ is an initial algebra.

\subsubsection{Widening the Type of $H$ and Adding Strength Laws to $\theta$}
Secondly, after this first step, we will still not be able to ensure the monad
laws for the obtained $(T,\eta,\prejoin)$ through an HSS. We are
lacking laws for $\th{}$ concerning its dependency on its second
argument. However, there is a third concern of a more practical
nature: if we take the first step, we have to define a prestrength
$\th{}$ for the signature functor $H$ generated in
\cref{sec:endofunctor} along the structure of the given multi-sorted
binding signature. As already mentioned there, the functor
$F^a: [\ccsort,\ccsort] \to [\ccsort,\CC]$, which is a building block
of that construction of $H$, does not have the same type as $H$, and
it is not even an endofunctor. So, if we are to develop a library of
functors with prestrength, we need to widen the notion so as to cover
those situation with ``varying types'' as well. We will then be able
to construct prestrengths for all building blocks of the signature
functor $H$. In parallel, we can propagate suitable laws for them
so that the laws governing the prestrength for $H$ ensure that an HSS
for it will satisfy the monad laws (see \cref{lem:monadsfromhsswithstrength} below).

\begin{definition}[\coqident{SubstitutionSystems.Signatures}{PrestrengthForSignature}]\label{def:prestrength}
  For categories $\CC$, $\D$, $\D'$ and a functor $H: [\CC,\D'] \to [\CC,\D]$, a \fat{prestrength} for $H$ is a natural transformation
  \[\th{}:(H {-}) \hcomp U {\sim} \arr H ({-} \hcomp U {\sim})\] between bifunctors
  $[\CC,\D'] \times \Ptd(\CC) \arr [\CC,\D]$.
\end{definition}

\begin{definition}[\coqident{SubstitutionSystems.Signatures}{StrengthForSignature}]\label{def:strength}
  Let a prestrength $\th{}$ be given for a functor $H$ as specified above.
  It is called a \fat{strength} for $H$ if it is ``homomorphic'' in the second
  argument $\sim$, in the following sense.
  The source and target bifunctors applied to a pair of objects $(X,(Z,e))$ with $X:[\CC,\D']$ and $(Z,e):\Ptd(\CC)$
  ($X$ for the argument symbolized by $-$ and $(Z,e)$ for the argument symbolized by $\sim$)
  yield $HX\hcomp Z$ and $H(X\hcomp Z)$, thus $\th{X,(Z,e)}:HX\hcomp Z\to H(X\hcomp Z)$ in $[\CC,\D]$.
  Being ``homomorphic'' of $\th{}$ in the second argument for us means satisfying the equations\footnote{See \coqident[theta_Strength1_int_nicer]{SubstitutionSystems.Signatures}{04e120e08a8d781ab7987bcb9001210b} and \coqident[theta_Strength2_int_nicest]{SubstitutionSystems.Signatures}{0d81eb5ac0beefdc398b56023a239436}.}
  $\th{X, \idwt{\Ptd(\CC)}} = \idwt{}_{HX}$ and
  \[  \th{X, (Z' \hcomp Z, e' \hcomp e)} = H(\alpha^{-1}_{X,Z',Z}) \vcomp   \th{X \hcomp Z', (Z,
      e)} \vcomp (\th{X, (Z', e')} \hcomp Z) \enspace , \]
  where $\alpha_{X,Y,Z} : X \hcomp (Y \hcomp Z) \simeq (X \hcomp Y) \hcomp Z$ is the associator.
  The equation is illustrated by the diagram in \cref{fig:sndstrengthlawsigs}.
  When $\th{}$ is a strength for $H$, the pair $(H,\th{})$ is called a \fat{signature with strength}.
  \begin{figure}[tb]
  \[
  \begin{xy}
   \xymatrix@C=0pt@R=25pt{
                  **[l]H X \hcomp Z' \hcomp Z  \ar[rr]^{\th{X,(Z' \hcomp Z, e' \hcomp e)}} \ar[d]_{\th{X, (Z', e')} \hcomp Z} & &  **[r]H(X \hcomp (Z' \hcomp Z)) \\
                  **[l] H(X \hcomp Z')\hcomp Z \ar[rr]^{\th{X \hcomp Z', (Z, e)}} & &   **[r]H((X \hcomp Z')\hcomp Z) \ar[u]_{H(\alpha^{-1}_{X,Z',Z})}
     }
  \end{xy}
\]
\vspace{-0.4cm}
\caption{Second strength law for signatures}\label{fig:sndstrengthlawsigs}
\vspace{-0.4cm}
\Description{}
\end{figure}
\end{definition}
Analogously to \cref{fig:othermonadlaws}, the upper left corner can be
associated either way for the type-checking of arrows out of it. On
the right-hand side, we cannot escape using the ``monoidal
isomorphism'' $\alpha$
witnessing associativity since the application of $H$ breaks the
convertibility of the types of arrows into $H(X \hcomp (Z' \hcomp Z))$
and $H((X \hcomp Z')\hcomp Z)$. Since $X$ is not assumed to be an endomorphism,
we use the notion of ``monoidal isomorphism'' by analogy only. Mathematically,
the arrow to the right is just the identity.
Approximatively, the diagram expresses that $\th{}$ with a composition in the second argument is the composition
of instances of $\th{}$ for each of the composites in the second argument.
This looks like a homomorphism of monoids, with $\Ptd(\CC)$ equipped with the identity
and the associative operation of composition, and the first equation
yields an identity, but the second equation varies also the first argument to $\th{}$.
In \cref{sec:understanding} we shed more light on this notion and in which sense this is ``homomorphic''.

In the case where $\D$ and $\D'$ are just $\CC$, all functor
compositions take place in $[\CC,\CC]$, and we precisely get back the
strength concept of
\citep{DBLP:journals/tcs/MatthesU04}, with the formalization-dictated
addition of the monoidal isomorphism that is part of the
formulation of \citep{DBLP:conf/types/AhrensM15}. In the formalization, we follow
\citep{DBLP:conf/types/AhrensM15} and have all monoidal isomorphisms
explicit. A contribution of the present paper is the
identification of the ``widening'' of the concept to
three possibly different parameter categories $\CC$, $\D$ and $\D'$,
as dictated by the necessity of a modular definition of the strength for
the signature functor generated by a \emph{multi-sorted} binding signature.
This formal widening is a conceptual step, but its implementation was direct
on the basis of the formalization provided by \citep{DBLP:conf/types/AhrensM15}.
Basically, all the proofs worked without changes other than the adaptation to the widening
of the statement. Such a ``refactoring'' step is thus painless given a full formalization.

It does not seem meaningful for us to form fixed points with signatures
with strength, unless the three parameter categories are identical:
the prominent role of pointed endofunctors on $\CC$ in the laws makes
them unapplicable for the analysis of other functors that might arise
as fixed points. So, otherwise, we consider those $H$ and their
strength only as building blocks.

As the notion of strength for $H$ appearing in the
definition of HSS is unaltered, we can simply use
\citep[Theorem 10]{DBLP:journals/tcs/MatthesU04}, formalized as
\citep[Theorem 26]{DBLP:conf/types/AhrensM15} that now reads as:
\begin{lemma}[\coqident{SubstitutionSystems.MonadsFromSubstitutionSystems}{Monad_from_hss}]\label{lem:monadsfromhsswithstrength}
  Let $\CC$ be a category with binary coproducts, and let $H$ be an endofunctor on
  $[\CC,\CC]$ and let $\th{}$ be such that $(H,\th{})$ is a signature with
  strength. If an $(\constfunctor{\Id{\CC}}+H)$-algebra 
  $(T,[\eta,\tau])$ is an HSS, then the definition
  \[\join \eqdef \gst{}{1_{(T,\eta)}} : T\hcomp T\to T \] yields a monad $(T,\eta,\join)$
  in monoid form.
\end{lemma}

In order to construct the strength from the building blocks of the generated signature functor $H$ in \cref{sec:endofunctor},
we will also use the concept of pointed distributive law from \citep[Definition 10]{signatures_to_monads}.
We avoid monoidal isomorphisms and give the following definition.
\begin{definition}[\coqident{SubstitutionSystems.SignatureExamples}{DistributiveLaw}]\label{def:distrlaw}
 Given $G:[\CC,\CC]$ for a category $\CC$, a \fat{pointed distributive law}
   is a natural transformation $\delta: G\hcomp U{\sim}\to
   U{\sim}\hcomp G$ of functors $\Ptd(\CC)\to[\CC,\CC]$ such that
  $\delta_\Id{\Ptd(\CC)} =\idwt{}_G$ and $\delta_{(Z' \hcomp Z, e' \hcomp e)}= Z'\hcomp\delta_{(Z,e)}\vcomp\delta_{(Z',e')}\hcomp Z$.\footnote{The second equation is \coqident[delta_law2_nicer]{SubstitutionSystems.SignatureExamples}{397731721c9f07e371ba59058f55b277}.}
\end{definition}

  We now sketch the construction of the strength for the obtained
  signature functor. For $\option_s$, construct a pointed distributive law, compose those distributive laws to
  obtain one for $\option~\ell$ and then generate the strength for
  precomposition with $\option~\ell$. Pointwise post-composition with a fixed functor
  (here with $\proj{t}$) generically
  allows us to construct a strength from the given one, which is good for $F^a$. This is a
  strength in the wide sense introduced in this paper.

  The product of finitely many such $F^{a_i}$ can be treated easily by
  using the corresponding construction of strengths for binary
  products iteratively. Let us mention the following implementation detail:
  the constructed signature with strength ought to have, as first
  component, the product of those $F^{a_i}$ w.\,r.\,t. \emph{convertibility}
  and not just provably. For one construction step, there is no
  concern, but this has to hold even for the iteration process,
  so the definitions have to go through the structure twice (as seen
  in \coqident{SubstitutionSystems.MultiSorted_alt}{Sig_exp_functor_list}).

The strength construction for $F^{(\vec{a},t)}$ is then dealt with by another instance of the strength construction
  for pointwise postcomposition with a fixed functor, this time $\hat t$.
  Canonically, the strength carries over to coproducts and thus to the
  $H$ associated to a multi-sorted binding
  signature.

  Following the steps described above, we construct a signature with
  strength from the given multi-sorted binding signature.  We have
  thus a full specification: starting from a multi-sorted binding
  signature, we construct a signature with strength. The only
  remaining task is to construct an
  $(\constfunctor{\Id{\CC}}+H)$-algebra representing terms that is, in particular, an
  HSS. This then provides us with $\prejoin$ as monad multiplication
  and thus a certified monadic substitution operation.

\section{From Multi-Sorted Binding Signatures to Certified Substitution}\label{sec:msbstocs}
Having constructed the signature with strength $(H,\theta)$ from a multi-sorted binding
signature, we are now ready to construct the monadic substitution
operation. The key theorem for this is:
\begin{theorem}[\coqident{SubstitutionSystems.MonadicSubstitution_alt}{TermMonad}]\label{thm:syntax-from-sig-with-strength}
  Let $(H,\theta)$ be a signature with strength, where
  $H : [\CC,\CC] \to [\CC,\CC]$ is \omegacocont and
  $\CC$ has binary coproducts, an initial object, and colimits of
  chains. Then the initial $(\constfunctor{\Id{\CC}}+H)$-algebra
  exists and is equipped with a monad structure.
\end{theorem}

The constructions in the proof of this theorem were given by
\citep{signatures_to_monads} and go via the heterogeneous
substitution systems of \cref{lem:monadsfromhsswithstrength}.
For the purposes of the present paper we treat
\cref{thm:syntax-from-sig-with-strength} as a black box, but stating
it this way makes it applicable to any abstract signature $(H,\theta)$
satisfying the assumptions. This way the result is decoupled from a
specific notion of concrete signature and applies both to the binding
signatures of \citep{signatures_to_monads} and those in
\cref{sec:mssigs}. It should also apply to other languages which are
not instances of these notions of signatures. An example is the
$\lambda$-calculus with explicit flattening as in
\citet{DBLP:journals/tcs/MatthesU04}, but \omegacocontinuity of this
has not yet been addressed properly and is left as an open question. We now turn
our attention to proving that $H$, as defined in
\cref{sec:endofunctor}, is \omegacocont.

\subsection{Proving \texorpdfstring{$\omega$-Cocontinuity}{omega-Cocontinuity} of the Signature Functor}\label{sec:omegacocont}

We start off by recalling the definition of \omegacocontinuity in
\UniMath. Colimits are parametrized by diagrams over graphs, as
suggested by \citet[p.~71]{maclane}. A functor $F : \CC \to \D$
\emph{preserves} colimits of shape $G$ if, for any diagram $d$ of
shape $G$ in $\CC$, and any cocone $a$ under $d$ with tip $C$, the
cocone $F\,a$ is colimiting for $F\,d$ whenever $a$ is colimiting for
$d$.

A functor is called \emph{\omegacocont} (\coqident{CategoryTheory.Chains.Chains}{is_omega_cocont}) if it preserves
colimits of diagrams of the shape
\[
  \begin{xy}
    \xymatrix@C=4pc{
      A_0 \ar[r]^{f_0} & A_1  \ar[r]^{f_1} & A_2  \ar[r]^{f_2} & \ldots
    }
  \end{xy}
\]
that is, diagrams on the graph where objects are natural numbers and
where there is a unique arrow from $m$ to $n$ if and only if
$1 + m = n$. We refer to diagrams of this shape as \emph{chains}.
Actually, in \UniMath, the type of arrows from $m$ to $n$ is
\emph{defined} to be the type of proofs that $1 + m = n$, exploiting
that the type of natural numbers is a set. 

In order to prove that the signature functor $H$ of
\cref{sec:endofunctor} satisfies \coqdocdefinition{is\_omega\_cocont}
we will apply lemmas that decompose the goal into smaller pieces. In
particular, the proof will boil down to proving that ``postcomposition
with $\proj{s}$'' and ``postcomposition with $\hat{t}$'' are
\omegacocont. Along the way we identify sufficient conditions
on $\CC$.

The general form of the following lemmas were added to \UniMath by
\citep{signatures_to_monads}. For simplicity we only recall them in
the form that we need them in this paper:

\begin{lemma}[Examples of \omegacocont functors]\label{lem:ex-pres-colim}\hfill
  \begin{enumerate}
    \item\label{pres-colim-const} Any constant functor
      $\constfunctor{d} : \CC \to \D$ is \omegacocont
      (\coqident{CategoryTheory.Chains.OmegaCocontFunctors}{is_omega_cocont_constant_functor}).
    \item\label{pres-colim-comp} The composition of
      \omegacocont functors is again \omegacocont
      (\coqident{CategoryTheory.Chains.OmegaCocontFunctors}{is_omega_cocont_functor_composite}).
    \item\label{theorem:times_cocont} Let $\CC$ and $\D$ be categories
      with specified binary products and further assume that
      $d\times{-}$ is \omegacocont for all $d : \D_0$. The
      binary product, $F \times G : \CC \to \D$, of
      \omegacocont functors $F,G : \CC \to \D$ is
      \omegacocont
      (\coqident{CategoryTheory.Chains.OmegaCocontFunctors}{is_omega_cocont_BinProduct_of_functors}).
    \item\label{lem:coprod_cocont} Let $\CC$ be a category and $\D$ a
      category with specified coproducts. Given an $I$-indexed family
      of functors $F_i : \CC \to \D$ the coproduct $\coprod_{i : I}
      F_i : \CC \to \D$ is \omegacocont
      (\coqident{CategoryTheory.Chains.OmegaCocontFunctors}{is_omega_cocont_coproduct_of_functors}).
    \item\label{theorem:precomp_cocont} Let $\CC$ have colimits of
      chains and let $F : \A \to \B$ be a functor, then the functor
      ``precomposition with $F$'', $\_ \hcomp F : [\B,\CC] \to
      [\A,\CC]$, is \omegacocont
      (\coqident{CategoryTheory.Chains.OmegaCocontFunctors}{is_omega_cocont_pre_composition_functor}).
  \end{enumerate}
\end{lemma}

By using point~\ref{lem:coprod_cocont} we can decompose
\omegacocontinuity of $H$ to showing that each $F^{(\vec{a},t)}$
is \omegacocont. By point~\ref{pres-colim-comp} and repeated
application of \ref{theorem:times_cocont} we now need to show
that $\hat{t} \hcomp \_$ and $F^{a_i}$ are
\omegacocont. Point~\ref{pres-colim-const} is used to handle
the case when $\vec{a}$ is empty. By points
\ref{pres-colim-comp} and \ref{theorem:precomp_cocont} the fact that
each $F^{a_i}$ is \omegacocont is reduced to \omegacocontinuity of
$\proj{s} \hcomp \_$.

As limits and colimits are computed pointwise in the functor category
$[\ccsort,\ccsort]$ we see that the assumptions on $\CC$ for making
the above argument go through is that $\CC$ has specified binary
products, coproducts, colimits of chains, and that $F \times{-}$ is
\omegacocont for all $F : [\ccsort,\CC]$.

The important facts that are not yet covered by the above argument are
the \omegacocontinuity of $\proj{t} \hcomp \_$ and $\hat{t} \hcomp \_$.
The key for proving these is:

\begin{lemma}[\coqident{CategoryTheory.Chains.OmegaCocontFunctors}{left_adjoint_cocont}]
  \label{lemma:left_adjoint_preserves_colimits}
  If $F : \CC \to \D$ is a left adjoint then it preserves colimits.
\end{lemma}

This classic result in category theory could also be found in \UniMath
and it implies that any left adjoint functor is
\omegacocont. However, \UniMath did not provide any lemmas
about the \omegacocontinuity of postcomposition with various
functors. We hence formalized the following lemma:

\begin{lemma}[\coqident{CategoryTheory.Adjunctions.Core}{is_left_adjoint_post_composition_functor}]
  If $F : \A \to \B$ is a left adjoint then $F \hcomp \_ : [\CC,\A]
  \to [\CC,\B]$ is a left adjoint.
\end{lemma}

Hence, if $F$ is a left adjoint, then $F \hcomp \_$ is \omegacocont
(\coqident{CategoryTheory.Chains.OmegaCocontFunctors}{is_omega_cocont_post_composition_functor}).
As remarked without proof in \cref{sec:endofunctor}, the functor
$\hat{t}$ is left adjoint to $\proj{t}$, so $\hat{t} \hcomp \_$ is
\omegacocont. We will now establish this adjunction as well as the
fact that $\proj{t}$ also has a right adjoint, which implies that
$\proj{t} \hcomp \_$ is \omegacocont.

In order to provide some intuition for the below proof, note
that $\proj{t}$ has a more abstract description. If we write $t$ also
for the map $1 \to \sort$ picking out the sort $t$, then
``precomposition with $t$'' is a functor $\_ \hcomp t : \ccsort \to
\CC^{1}$. It is easy to see that this is isomorphic to
$\proj{t}$, so if we assume that $\CC$ has suitable (co)limits
then it will have both left and right adjoints given by Kan
extension. This abstract argument indicates that we should indeed be
able to establish that $\proj{t}$ has both left and right adjoints
given that $\CC$ has suitable (co)limits. However, we can characterize
these much more concretely as follows:

\begin{lemma}[\coqident{SubstitutionSystems.MultiSorted_alt}{is_left_adjoint_hat}]\label{lem:hat}
  Let $\CC$ be a category with specified set-indexed coproducts. The
  functor $\proj{t}$ has a left adjoint $\hat{t} : \CC \to \ccsort$
  defined as
  \[
    \hat{t}~X~s \eqdef \coprod_{(t = s)} X\enspace.
  \]
\end{lemma}
\begin{proof}
  Since $\sort$ is a discrete category, the set $\hom_{\ccsort}(\hat{t}~A,B)$ is
  just a family of morphisms from $\hom_\CC(\hat{t}~A~s,B~s)$, for any
  $s : \sort$. Unfolding the definition of $\hat{t}$ gives us
  $\hom_\CC(\coprod_{(t = s)} A,B~s)$ which is nontrivial only when
  $t = s$. On the other hand
  $\hom_\CC(A,\proj{t}~B) \convert \hom_\CC(A,B~t)$, so these sets are
  clearly isomorphic.
\end{proof}

\begin{lemma}[\coqident{SubstitutionSystems.MultiSorted_alt}{is_left_adjoint_projSortToC}]\label{lem:underline}
  Let $\CC$ be a category with specified products. The functor
  $\proj{t}$ has a right adjoint $\underline{t} : \CC \to \ccsort$
  defined as
  \[
    \underline{t}~A~s \eqdef \prod_{(t = s)} A\enspace.
  \]
\end{lemma}
\begin{proof}
  This is dual to the proof of \cref{lem:hat}.
\end{proof}

The reason we require the coproducts to be set-indexed is that this is
needed for the existence of $H$ (recall that $I$ is assumed to be a
set). For \cref{lem:underline} it would have sufficed to assume
existence of proposition-indexed products, however, the stronger
assumption of arbitrary (small) products is usually not a problem as the
existence of such products in concrete categories typically does not depend
on the homotopy level of the indexing type.
For instance, $\set$ has products indexed by
arbitrary (small) types, but only set-indexed coproducts.
Combining all of this we get:

\begin{theorem}[\protect{\coqident[omega_cocont_MultiSortedSigToSig]{SubstitutionSystems.MultiSorted_alt}{is_omega_cocont_MultiSortedSigToSignature}}]\label{theorem:sig_functor_cocont}
  Given a category $\CC$ with chosen products, set-indexed coproducts,
  colimits of chains, such that $F \times{-}$ is \omegacocont
  for all $F : [\ccsort,\CC]$, then the signature functor $H$
  associated to a binding signature is \omegacocont.
\end{theorem}

In the formalization we provide a slightly different interface to the
construction than the one of \cref{theorem:sig_functor_cocont}. In
particular, we assume that $\CC$ has chosen terminal and initial
objects as well as binary (co)products. The reason for this is that,
even though these assumption follow from the assumptions of the
theorem, one can often give more direct constructions of these
assumptions in concrete categories.

\subsection{Instantiating the Framework and Examples}
\label{sec:examples}

We can now instantiate \cref{theorem:sig_functor_cocont} and then
\cref{thm:syntax-from-sig-with-strength} with $\CC = \set$. The
\UniMath library contains proofs that $\set$ is (co)complete as well
as lemmas about (co)limits in functor categories, so most assumptions
are easily satisfied. Proving that $\set$ has colimits is in fact one
of the places where we profit from working in univalent type theory.
This construction relies on set quotients which are not directly
available in intensional type theory, for details see
\cite[Section~3.3]{signatures_to_monads}. The last assumption of
\cref{theorem:sig_functor_cocont} is satisfied if $[\ccsort,\CC]$ has
exponentials. These are not as easy to compute in functor categories
as (co)limits, luckily \UniMath already has a formalization of them
when the target category is $\set$. We can hence instantiate the
framework and obtain some examples.

\begin{example}[\coqfile{}{SubstitutionSystems.STLC_alt}]
  Assuming that $\sort$ is closed under a binary operation
  $\Rightarrow : \sort \to \sort \to \sort$, we have the multi-sorted
  signature of STLC from \cref{example:stlc}. Applying the
  construction of \cref{sec:endofunctor} to get $H$ as in
  \cref{ex:sigSTLC}, and providing it, together with its strength
  $\theta$, to \cref{thm:syntax-from-sig-with-strength}, we obtain an
  initial algebra $\Lambda : [\set^{\sort},\set^{\sort}]$ for $H$
  together with a monadic substitution operation. Given
  $X, Y : \set^{\sort}$ and a morphism $f : X \to \Lambda~Y$, the
  substitution operation yields a morphism
  $\cfont{subst}~f : \Lambda~(X + Y) \to \Lambda~Y$. This is a
  parallel substitution, replacing all occurrences of $X$ in one
  go. By instantiating with $1$ for $X$ we obtain an
  operation which just replaces the occurrence of one variable.

  The monad laws state that $\cfont{subst}$ is well-behaved and from
  them we can derive various standard laws for substitution. For
  example, the interchange law for commuting two substitutions can be
  derived generically for any monad $M$. By instantiation we then
  obtain that $\cfont{subst}$ satisfies this law. For details, see the
  formal development.
\end{example}

We have also implemented more complex languages, including Plotkin's
PCF~\citep{Plotkin77} and the pre-syntax of the Calculus of
Constructions (CoC) à la \citet{DBLP:books/daglib/0067012}. For space restrictions
we refer to the formalization for details.

\begin{example}[\coqfile{}{SubstitutionSystems.PCF_alt}]
  \label{ex:pcf}
  PCF is an extension of STLC with natural number and boolean
  types, together with operations on these. Specifying it as a
  multi-sorted binding signature is direct and uses the disjoint sum
  of signatures to extend the signature of STLC.
\end{example}

\begin{example}[\coqfile{}{SubstitutionSystems.CCS_alt}]
  \label{ex:coc}
  In \citep{DBLP:books/daglib/0067012} the CoC is presented as a 2-sorted language with
  $\sort = \{ \cfont{ty}, \cfont{el} \}$ representing the types and
  terms. The pre-syntax part of this is a simple multi-sorted binding
  signature, with binders for both
  types and terms. Formalizing it posed no difficulties.
\end{example}

\section{Understanding the Notion of Strength}\label{sec:understanding}
Originally, the aim of the work presented in this paper was to benefit
from our previous construction of a certified monadic
substitution for \emph{untyped} terms based on \omegacocont base functors $H$ with
strength, which is also implemented in \UniMath~\citep{signatures_to_monads}.
The move is to a more complex
base category to account for \emph{multi-sorted} terms, the construction of $H$ now
from a \emph{multi-sorted} binding signature.
But for the construction
of strength of $H$ to be modular, we identified the need for the
widening of the strength concept to account for functors that are not endo.
This all worked, in the sense of providing computer-formalized
mathematical structures, with a certified monadic substitution
operation as outcome. The original strength concept had already been
understood mathematically, in the sense of relating it to more
abstract mathematical notions, by \citet{DBLP:conf/types/AhrensM15}:
they propose a definition of \emph{relative strength} \citep[Definition
11]{DBLP:conf/types/AhrensM15} that is spelt out on the level of
monoidal categories, but directed towards the implemented strength.
However, the authors sketched the instantiation very briefly and also
mentioned that this is an instance of a strength concept for actions.
The definition of relative strength just mentioned (``relative'' does not suggest
that we are aiming for \emph{relative monads}), and also the use of action
strength by \citet{DBLP:conf/lics/Fiore08}, were geared towards uses
with a specific monoidal category, the one canonically obtained through composition of endomorphisms.

In what follows, we use two other dimensions of \UniMath beyond the
library of certified structures: we are formalizing the meta-theory
behind the concepts, in particular, we prove the claims in
\citep{DBLP:conf/types/AhrensM15} as to the link with abstract
strength concepts. We also use it as a research tool to obtain new
results on the mathematical justification of the widened strength
notion, here by \begin{enumerate*}[label=(\roman*)] \item identifying a bicategorical scenario for action
strength, and by \item continuing the investigation into a ``higher-order''
view of actions as monoidal functors, as put forward by \citet{JanelidzeKelly2001}. \end{enumerate*} We
even mention work in progress in that spirit that sees action
strength itself as a monoidal functor.

\subsection{Action-Based Strength}\label{sec:actionbased}

Action-based strength will be our term for the extra requirement on a
functor $F:\A\to\A'$ to be a morphism between actions of a monoidal
category $\V$ on $\A$ and $\A'$, respectively. We first present the
``point-free version''
(\coqident{CategoryTheory.Monoidal.ActionBasedStrength}{param_distr_pentagon_eq_body})
based on the definition of actions found in \citep{JanelidzeKelly2001}.
This version has the advantage of asking for very little to be added to ``general
category theory''. However, the ``natural'' level of generality of its definition
is in an arbitrary bicategory instead of the different functor categories, still with an ordinary monoidal
category as parameter space, and most of our formal proof development is on that level.

We consider a monoidal category $\V$ with tensor product $\otimes$,
categories $\A$ and $\A'$ and actions of $\V$ on both, expressed as
strong monoidal functors $F:\V\to[\A,\A]$ and $F':\V\to[\A',\A']$
(where the monoidal structure on the endofunctor categories is given by 
composition in diagrammatic order). Let $G:\A\to\A'$ be a functor and let $\delta$ be a natural
transformation between functors from $\V$ to $[\A,\A']$ that
are defined on objects $v$ of $\V$ as $F'v\cdot G$ and $G\cdot Fv$,
respectively. The transformation $\delta$ is a \fat{parameterized
  distributivity} for $G$ (but could also be called a
\emph{strength} of $G$) if the diagrams in \cref{fig:actionbasedstrengthalt} commute.

\begin{figure}[tb]
  \[ \def\arraystretch{2.2}
    \begin{array}{c}
    \xymatrix@C3pc@R1pc{
      F'I\cdot G \ar[rr]^{\delta_I}&&G\cdot FI\\
      &G \ar[ul]^{\epsilon'\cdot G}\ar[ur]_{G\cdot\epsilon}
        }\\
\xymatrix@C-1.5pc@R-1pc{
                  **[l] F'(v\otimes w)\cdot G  \ar[rr]^{\delta_{v\otimes w}}  & &  **[r] G\cdot F(v\otimes w) \\
                  **[l] F'w\cdot F'v\cdot G \ar[u]^{\mu'_{v,w}\cdot G}\ar[dr]_-{F'w\cdot \delta_v\hspace{1em}} & &   **[r]G\cdot Fw\cdot Fv \ar[u]_{G\cdot \mu_{v,w}}\\
                  &   **[c] F'w\cdot G\cdot Fv \ar[ur]_-{\hspace{1em}\delta_w\cdot Fv}
}
\end{array}
\]
\vspace{-0.4cm}
\caption{\coqident{CategoryTheory.Monoidal.ActionBasedStrength}{param_distr_triangle_eq}, \coqident{CategoryTheory.Monoidal.ActionBasedStrength}{param_distr_pentagon_eq_body}}\label{fig:actionbasedstrengthalt}
\vspace{-0.4cm}
\Description{}
\end{figure}

In the diagrams, all nodes are functors from $\A$ to $\A'$, with $I$ the unit of $\V$. The arrows $\epsilon:\Id{}\to FI$ and
$\mu_{v,w}:Fw\cdot Fv\to F(v\otimes w)$ are the isomorphisms coming
from $F$ being strong monoidal, and likewise for $\epsilon'$ and
$\mu'$ relative to $F'$. Since the vertically arranged morphisms in the pentagon
diagram are isomorphisms, the diagram is essentially a triangle
describing $\delta_{v\otimes w}$ as a suitable composition of
$\delta_v$ and $\delta_w$, as expected (and in accordance with our notion of strength in \cref{def:strength}).
Let us remark that $(G,\delta)$ cannot be seen as a monoidal natural transformation from $F$ to $F'$---already
because they need not have the same target category---so this definition
is specific to the interpretation of the strong monoidal functors $F$ and $F'$ as actions.

While the point-free definition is concise, we rather base our analysis of
strength for signatures on a more concrete version that is very close
to the original definition of \citet{Pareigis1977}, coming under the
name of $\CC$-categories for actions and $\CC$-functors for strong
functors between actions. We even formalized a transformation of
action strengths from the point-free to the concrete setting
(\coqident{CategoryTheory.Monoidal.ActionBasedStrength}{actionbased_strong_functor_from_alt}).
The concrete version exploits the equivalence between the functor
categories $[\A\times \B,\CC]$ and $[\B,[\A,\CC]]$ which is currying /
uncurrying (together with swapping $\A$ and $\B$). This equivalence extends to a
biadjunction in the bicategory of (small) categories
(\coqident{Bicategories.PseudoFunctors.Examples.CurryingInBicatOfCats}{currying_biajd}).

\begin{definition}[\coqident{CategoryTheory.Monoidal.Actions}{action}]\label{def:action}
An \fat{action} of the given monoidal category $\V$
(as above) is a functor $\odot:\A\times\CC\to\A$ (instead of its
curried version $F:\CC\to[\A,\A]$) that is written infix, together
with natural isomorphisms $\varrho$, with components $\varrho_a:a\odot I\to a$ (\emph{right unitor} of the action), and
$\chi$, with components $\chi_{a,v,w}:(a\odot v)\odot w\to a\odot(v\otimes w)$ (\emph{convertor} of the action), that make the
diagrams in \cref{fig:action} commute (with $\lambda$ and $\alpha$ the left unitor and associator of $\V$, respectively).
\end{definition}
\begin{figure}[tb]
\[ \def\arraystretch{2.2}
  \begin{array}{c}
\xymatrix@C3pc@R1pc{
    (a\odot I)\odot v \ar[rr]^{\chi_{a,I,v}} \ar[dr]_{\varrho_a\odot 1_v}& & a \odot (I \otimes v) \ar[dl]^{1_a\odot \lambda_v} \\
    & a \odot v
      }\\
 \xymatrix@C-2pc{
    (a\odot u)\odot(v\otimes w) \ar[rr]^{\chi_{a,u,v\otimes w}}& & a\odot(u\otimes(v\otimes w)) \\
    ((a\odot u)\odot v)\odot w \ar[u]_{\chi_{a\odot u,v,w}} \ar[dr]_{\chi_{a,u,v}\odot 1_w}& &a\odot((u\otimes v)\otimes w) \ar[u]^{1_a\odot\alpha_{u,v,w}}\\
    &(a\odot(u\otimes v))\odot w \ar[ur]_{\chi_{a,u\otimes v,w}}
      }
\end{array}
\]
\vspace{-0.4cm}
\caption{\coqident{CategoryTheory.Monoidal.Actions}{action_triangle_eq}, \coqident{CategoryTheory.Monoidal.Actions}{action_pentagon_eq}}\label{fig:action}
\vspace{-0.6cm}
\Description{}
\end{figure}

\begin{definition}[\coqident{CategoryTheory.Monoidal.ActionBasedStrength}{actionbased_strength}]\label{def:actionbasedstrength}
  Let $\V$ be a mono\-idal category, and let $\odot$ and $\odot'$ be $V$-actions on categories $\A$, $\A'$, with $\varrho$, $\varrho'$, $\chi$, $\chi'$ the right unitors and convertors of the actions, respectively.
  Let $G:\A\to\A'$ be a functor. A natural transformation $\zeta$, with components $\zeta_{a,v}:Ga\odot' v\to G(a\odot v)$ is an \fat{action-based strength} for $G$ if the diagrams in \cref{fig:actionbasedstrength} commute. The pair $(G, \zeta)$ is then called a \fat{strong action-based functor} from $\odot$ to $\odot'$.
\end{definition}
 \begin{figure}[tb]
   \[ \def\arraystretch{2.2}
  \begin{array}{c}
\xymatrix@C3pc@R1pc{
    Ga\odot' I \ar[rr]^{\zeta_{a,I}} \ar[dr]_{\varrho'_{Ga}}& & G(a \odot I) \ar[dl]^{G\varrho_a} \\
    & Ga
      }\\
 \xymatrix@C-1.5pc{
    Ga\odot'(v\otimes w) \ar[rr]^{\zeta_{a,v\otimes w}}& & G(a\odot(v\otimes w)) \\
   (Ga\odot' v)\odot' w\ar[u]_{\chi'_{Ga,v,w}} \ar[dr]_{\zeta_{a,v}\odot' 1_w}& &G((a\odot v)\odot w) \ar[u]^{G\chi_{a,v,w}}\\
    &G(a\odot v)\odot' w \ar[ur]_{\zeta_{a\odot v,w}}
      }
\end{array}
\]
\vspace{-0.4cm}
\caption{\coqident{CategoryTheory.Monoidal.ActionBasedStrength}{actionbased_strength_triangle_eq}, \coqident{CategoryTheory.Monoidal.ActionBasedStrength}{actionbased_strength_pentagon_eq} (correspond to those of \cref{fig:actionbasedstrengthalt})}\label{fig:actionbasedstrength}
\vspace{-0.4cm}
\Description{}
\end{figure}

Notice that \citep{Pareigis1977} required $\zeta$ to be a natural isomorphism, which would make our application to
multi-sorted binding signatures impossible (also \citet{DBLP:conf/lics/Fiore08} worked without the requirement of isomorphism).
Already in \citep{Pareigis1977}, there is a notion of morphism between strong action-based functors, under the name $\CC$-morphism. With our notations:
\begin{definition}[Natural transformation of strong action-ba\-sed functors]\label{def:nattransabstrong}
  Given the data in the preceding definition and a further functor
  $G':\A\to\A'$ with action-based strength $\zeta'$, a natural
  transformation $\eta:G\arr G'$ is an \fat{action-based natural
    transformation} if the diagram
\begin{equation}
\xymatrix@C3pc@R1pc{
  Ga\odot'v  \ar[r]^{\zeta_{a,v}} \ar[d]_{\eta_a\odot' 1_v}& G(a\odot v)\ar[d]^{\eta_{a\odot v}}\\
  G'a\odot'v  \ar[r]^{\zeta'_{a,v}}& G'(a\odot v)
      }
\label{fig:actionnattrafo}
\end{equation}
commutes (\coqident{CategoryTheory.Monoidal.ActionBasedStrongFunctorCategory}{Strong_Functor_Category_mor_diagram}).
\end{definition}
For a fixed monoidal category $\V$, categories $\A$ and $\A'$ and
actions $\odot$ and $\odot'$ of $\V$ on $\A$ and $\A'$, we define the
``action-based strong functor category'' with the strong action-based
functors from $\odot$ to $\odot'$ as objects and the action-based
natural transformations between those objects as morphisms. This is
comfortably done in \UniMath by exploiting the concept of displayed
category introduced by \citet{DBLP:journals/lmcs/AhrensL19}, meaning
that only the extra data and constraints with respect to the
underlying functor category $[\A,\A']$ have to be given. This leads to
a displayed category
(\coqident{CategoryTheory.Monoidal.ActionBasedStrongFunctorCategory}{Strong_Functor_category_displayed}),
from which the sought category is obtained generically as the
``total'' category
(\coqident{CategoryTheory.Monoidal.ActionBasedStrongFunctorCategory}{Strong_Functor_category}).
In this case, a formalized form of the ``structure identity
principle'' \cite[Section 9.8]{hottbook} allows us to prove in a
straightforward and modular way that the category thus obtained is
univalent when the target category $\A'$ is (\coqident{CategoryTheory.Monoidal.ActionBasedStrongFunctorCategory}{is_univalent_Strong_Functor_category}).
For this to be applicable, it is crucial that the notion of \emph{action-based} natural
transformation does not come with
supplementary data, just with a property (\ie, a proposition) to be satisfied. This
makes it especially rewarding to use the framework of displayed
bicategories by \citet{DBLP:conf/rta/AhrensFMW19} (as implemented in
\coqfile{}{Bicategories.DisplayedBicats.DispBicat}) to construct the
bicategory of actions, strong action-based functors and action-based
natural transformations (for a given monoidal category $\V$) as displayed over the bicategory of (small)
categories
(\coqident{CategoryTheory.Monoidal.ActionsFormBicategory}{actions_disp_bicat})
(from which the library immediately derives an ``ordinary'' bicategory
(\coqident{CategoryTheory.Monoidal.ActionsFormBicategory}{actions_bicat})).
Also notice that the diagram in \cref{fig:actionnattrafo} is
symmetric in the sense that if $\eta$ is an isomorphism, its inverse
satisfies it again (with $(G,\zeta)$ and $(G',\zeta')$ exchanged). In other words, the
obtained displayed bicategory is locally a groupoid
(\coqident{CategoryTheory.Monoidal.ActionsFormBicategory}{actions_disp_locally_groupoid}).
We are not aware of such an observation in the literature.

\subsection{Signature Strength as Action-Based Strength}\label{sec:instabstrength}

In the interest of securing the meta-theory as well, we formalize relative strength of \citep[Definition
11]{DBLP:conf/types/AhrensM15} (\coqident{CategoryTheory.Monoidal.ActionBasedStrength}{rel_strength}) and formally confirm
the claims of \citep{DBLP:conf/types/AhrensM15}:
\begin{enumerate*}[label=(\roman*)]
\item The parameters of action-based strength can be instantiated so
  that one can construct transformations from relative strength to
  instantiated action-based strength
  (\coqident[from_relative_strength]{CategoryTheory.Monoidal.ActionBasedStrength}{actionbased_strength_from_relative_strength})
  and vice versa
  (\coqident[from_actionbased_strength]{CategoryTheory.Monoidal.ActionBasedStrength}{relative_strength_from_actionbased_strength}).
\item For the case of $\D=\D'=\CC$ in \cref{def:strength}, the
  parameters of relative strength can be instantiated so that one can
  construct transformations from strength to instantiated relative
  strength\label{item:instrel}
  (\coqident[from_signature]{SubstitutionSystems.Signatures}{rel_strength_from_signature})
  and vice versa
  (\coqident[signature_from]{SubstitutionSystems.Signatures}{signature_from_rel_strength}).
\end{enumerate*}

The endofunctors, and the pointed endofunctors, over a category $\CC$
readily form monoidal categories, and the forgetful functor $U$ that
forgets the ``points'' is even a strong monoidal functor between these
monoidal categories
(\coqident[from_ptd_as_strong_monoidal_functor]{CategoryTheory.Monoidal.PointedFunctorsMonoidal}{forgetful_functor_from_ptd_as_strong_monoidal_functor}).
It is this $U$ that is used in step~\ref{item:instrel} above to
instantiate the parameter of that name in relative strength (mentioned
already in \citep{DBLP:conf/types/AhrensM15}).
However, in the notion of
prestrength, we also have to deal with functors that are not endos.
The notion of action is wide enough to handle
this: the endomorphisms of $\CC$ act on any functor category with
source category $\CC$, see below. To see the situation of \cref{def:prestrength} more
clearly, we abstract away from the bicategory of (small) categories
with their functors and natural transformations and assume a
bicategory $B$. The categories $\CC$, $\D$ and $\D'$ of
\cref{def:prestrength} are replaced by objects $c_0$, $d_0$, $d_0'$ of
$B$. The object $c_0$ gives rise to a monoidal category $\Endo_B(c_0)$ built from
the hom-category $B(c_0,c_0)$, \ie, the endomorphisms of $c_0$ and
their 2-cells. $\Endo_B(c_0)$ replaces $[\CC,\CC]$ in the inner workings of \cref{def:prestrength}, while
$\Ptd(\CC)$ is replaced by an arbitrary monoidal category $\V$, and with an arbitrary strong monoidal $U$ from $\V$
to $\Endo_B(c_0)$.

Now, for any objects $c_0,d_0$ of $B$, we can define an action $\odot_{c_0,d_0}$ of $\Endo_B(c_0)$
on the hom-category $B(c_0,d_0)$, which on objects $f:c_0\to d_0$, $g:c_0\to c_0$ is given by $f \odot_{c_0,d_0}g \eqdef f\cdot g$.
We say that $\Endo_B(c_0)$ \emph{acts by precomposition} on $B(c_0,d_0)$ (together with the right unitor and
convertor satisfying the laws in \cref{fig:action}),
formalized as \coqident{CategoryTheory.Monoidal.ActionOfEndomorphismsInBicat}{action_from_precomp}.

Given a strong monoidal functor from $\V$ to $\V'$, any action of $\V'$ can be lifted to an action of $\V$ on
the same category (this generalizes an observation by \citet[p.~59]{DBLP:conf/lics/Fiore08} that the (obvious)
action of $\V'$ on itself can be lifted that way),
formalized as \coqident{CategoryTheory.Monoidal.ConstructionOfActions}{lifted_action}.
Thus, we get an action of $\V$ on $B(c_0,d_0)$. Reusing this action construction for $d_0'$ in place of $d_0$ as well,
one can study action-based strength for a given functor
$G:B(c_0,d_0')\to B(c_0,d_0)$ that replaces $H$ of \cref{def:prestrength}. The laws of \cref{fig:actionbasedstrength}
can then be instantiated to the present abstract bicategorical scenario
(see the logically equivalent formalized \coqident{CategoryTheory.Monoidal.ActionBasedStrengthOnHomsInBicat}{triangle_eq_nice} and
\coqident{CategoryTheory.Monoidal.ActionBasedStrengthOnHomsInBicat}{pentagon_eq_nice}).

Of course, we want to instantiate all of this to the bicategory $\CAT$ of (small)
categories, formalized as
\coqident{Bicategories.Core.Examples.BicatOfCats}{bicat_of_cats}.
A small technical problem arises as the forgetful functor
$U$ from $\Ptd(\CC)$ is not seen by the system as having target
$\Endo_\CAT(\CC)$. Once this is solved, we identify the equivalence of
this instance of action-based strength with the notion of strength in
\cref{def:strength}. This equivalence is not only logical, but extends
to an adjoint equivalence of a suitably formed category of signatures
with strength (for given $\CC$, $\D$, $\D'$) and the instance of the
action-based strong functor category mentioned after
\cref{def:nattransabstrong}, formalized as
\coqident{CategoryTheory.Monoidal.ActionBasedStrengthOnHomsInBicat}{EquivalenceSignaturesABStrongFunctors}.

\subsection{Understanding Strength Itself}\label{sec:strengthitself}

In category theory, many concepts can be mutually reduced to each
other, by appropriately chosen instances for the parameter categories
of these concepts. \citet{maclane} has many examples already for the
most common concepts such as initial objects, limits and adjunctions.
It is a matter of taste if an action is better ``understood'' in form
of the concrete \cref{def:action} or as a monoidal functor into
endomorphisms, but the second concept had already been there. For
action strength, one might ask if the extra concept of parameterized
distributivity has some merit over the likewise concrete, but prior
\cref{def:actionbasedstrength} (even if the concepts relate through categorical
currying, we consider them as distinct).

The vertical arrows in \cref{fig:actionbasedstrengthalt} are
isomorphisms, as are the legs of the triangle, so one might be tempted
to say that the laws express that $\delta$ is a homomorphism, with
$\delta_I$ being the ``identity'', and $\delta_{v\otimes w}$ being the
``composition'' of $\delta_v$ and $\delta_w$. However, the nodes in
the pentagon diagram depend on the objects $v$, $w$ of $\V$, and
$\delta$ is a natural transformation between functors into a functor
category.
Since the parameters run through the monoidal category $\V$, the
appropriate notion of homomorphism for $\delta$ is necessarily by way
of a monoidal functor with source $\V$ to which $\delta$ closely
corresponds.

We exploit that natural transformations can sometimes be seen as
functors and extend this to the representation of parameterized
distributivity as strong monoidal functors.

First, we study the elementary situation in plain categories, with a
precise connection between, on the one hand, natural transformations of functors into a
functor category and, on the other hand, functors into a specifically crafted category. More
precisely, natural transformations between $H,H':\CC\to[\A,\A']$
correspond to functors from $\CC$ to a category
$\T(H,H')$
(\coqident{CategoryTheory.Monoidal.ActionBasedStrongFunctorsMonoidal}{trafotarget_cat})
that is the total category of a displayed category over $\CC$
(\coqident{CategoryTheory.Monoidal.ActionBasedStrongFunctorsMonoidal}{trafotarget_disp}):
an object ``over'' $c$ is the pair consisting of $c$ and a natural
transformation $\alpha^c:H\,c\to H'\,c$---hence a morphism of $[\A,\A']$---and
a morphism ``over'' $f:c\to c'$ is a pair consisting of
$f$ and a proof that the ``naturality diagram'' commutes:
\begin{equation}
\xymatrix@C3pc@R1pc{
  Hc  \ar[r]^{\alpha^c} \ar[d]_{Hf}& H'c\ar[d]^{Hf'}\\
  Hc'  \ar[r]^{\alpha^{c'}}& H'c'
      }
\label{fig:trafotarget}
\end{equation}
This diagram does not ask for naturality of a transformation $\alpha:H\to H'$, it is only
expressed between the components $\alpha^c$ and $\alpha^{c'}$, and this only for morphism $f:c\to c'$.
When running through all objects and all morphisms of $\CC$, the ingredients added to $\CC$ in order to obtain  $\T(H,H')$
constitute such a natural transformation $\alpha:H\to H'$.

There is a forgetful
functor $U:\T(H,H')\to\CC$, and already from the previous description, it becomes clear that the natural transformations $\alpha$
from $H$ to $H'$ are in bijection with functors $N$ from $\CC$ to
$\T(H,H')$ that satisfy $U\cdot N=1_\CC$. This latter
identity is between functors, and, if expressed naïvely, relies on equality of objects in a category. In particular, the type of such identities is not a proposition.
The way out is to exploit that the target category is obtained through a displayed category. There is an elementary characterization
of functors from a source category into the total category of a displayed category for which composition with the forgetful functor agrees with
a given functor into the base of that target category, which in our case is just $1_\CC$. The elementary concept here is the formalization
of the notion of \emph{section}, in particular \coqident{CategoryTheory.DisplayedCats.Core}{section_disp}.
To show the above-mentioned correspondence, we thus concretely formalized a bijection between the natural
transformations from $H$ to $H'$ and the sections with respect to $\CC$ and the displayed category underlying $\T(H,H')$ (\coqident{CategoryTheory.Monoidal.ActionBasedStrongFunctorsMonoidal}{nat_trafo_to_section} and \coqident{CategoryTheory.Monoidal.ActionBasedStrongFunctorsMonoidal}{section_to_nat_trafo}).

It turns out that these observations can be smoothly carried over to an arbitrary bicategory for the elements in the target, instead of limiting them to functors and natural transformations. This gives a formalized bijection in form of \coqident{CategoryTheory.Monoidal.ActionBasedStrongFunctorsMonoidal}{nat_trafo_to_section_bicat} and \coqident{CategoryTheory.Monoidal.ActionBasedStrongFunctorsMonoidal}{section_to_nat_trafo_bicat}.

In the situation of parameterized distributivity, we are lacking the
notion of displayed monoidal category and the suitable section
characterization for strong monoidal functors from a given monoidal
category into the total category of the displayed monoidal category
whose composition with the forgetful functor yields a given strong
monoidal functor from the source into the base. We leave the development of such a notion as an open problem. It is essentially not a mathematical problem but
a formalization challenge---this time with the purpose of avoiding reasoning with
equality between strong monoidal functors. In this paper, we can only offer
one direction of the correspondence: given a parameterized distributivity, we construct
a strong monoidal functor from the monoidal category $\V$ of parameters into a specially crafted monoidal category.

We use the general result above and instantiate $\CC$
with $\V$ and $H$ and $H'$ with the source and target of $\delta$
and
add monoidal structure to $\T(H,H')$. Recall that objects of the category $\T(H,H')$ contain natural transformations from $\A$ to $\A'$ and that morphisms
of $\T(H,H')$ come with proof obligations for equations between such natural transformations. This allows us to
precisely capture the diagrams of \cref{fig:actionbasedstrengthalt} in the definition of the objects of $\T(H,H')$ that are needed
to turn it into a monoidal category: For the unit, we add
to the unit $I$ of $\V$ the natural transformation
$(G\cdot\epsilon)\vcomp((\epsilon')^{-1}\cdot G)$ to which $\delta_I$ is
supposed to be equal. In order to get the tensor on objects $(v,dv)$ and $(w,dw)$ of $\T(H,H')$,
we add to $v\otimes w$ the natural transformation that arises from
taking in the pentagon diagram the path that is meant to correspond to
$\delta_{v\otimes w}$, with $\delta_v$ and $\delta_w$ replaced by $dv$
and $dw$, respectively. This extends to a tensor operation on
$\T(H,H')$, but not without effort: we ask for the morphism of $\T(H,H')$ that corresponds to the tensor of two such morphisms, thus we have
to prove an equality between two chains of compositions of natural transformations. And there are more morphisms of $\T(H,H')$ asked for
in order to get left and right unitor and the associator. The formalization effort for these equalities turns out to be incommensurate with their proofs
on paper, the main problem being that there is too much structure in those ``concrete'' functors and natural transformations, so that
any attempt at simplifying the goals at hand makes them unreadably convoluted. The way out is a bicategorical generalization of the problem at hand,
extending the bicategorical generalization of the correspondence for the case of plain categories we mentioned above.
This does not mean that the formalization becomes an easy task, but there is still an adequation between the pencil-and-paper proof and the formalization.
Moreover, a part of the proofs has been obtained without a prior pencil-and-paper proof, which is a major desideratum for the development of formalized mathematical
results. Here, we briefly mention some of the more important elements of the formalization:
\coqident{CategoryTheory.Monoidal.ActionBasedStrongFunctorsMonoidal}{montrafotargetbicat_cat} is the target category, which is extended to a monoidal category
\coqident{CategoryTheory.Monoidal.ActionBasedStrongFunctorsMonoidal}{montrafotargetbicat_moncat}, where the morphism part of the tensor rests on
Lemma \coqident{CategoryTheory.Monoidal.ActionBasedStrongFunctorsMonoidal}{montrafotargetbicat_tensor_comp_aux} with a proof of over 200 lines,
and six more such lemmas for the unitors and the associator whose construction in total requires more than 1kloc. The laws themselves do not need much effort
since the extra data for morphisms in $\T(H,H')$ consists of equalities, hence their equalities are trivial under our general assumption that the
homsets of all categories in our formalization are sets in the sense of univalent foundations.
The main result is the bicategorical generalization of the translation of parameterized distributivity into strong monoidal
functors \coqident{CategoryTheory.Monoidal.ActionBasedStrongFunctorsMonoidal}{smf_from_param_distr_bicat_parts}
(the notion of parameterized distributivity is not yet properly generalized to bicategories), and the main result of this section is then the
instantiation of the monoidal target category
\coqident{CategoryTheory.Monoidal.ActionBasedStrongFunctorsMonoidal}{montrafotarget_moncat} and of the translation
itself \coqident{CategoryTheory.Monoidal.ActionBasedStrongFunctorsMonoidal}{smf_from_param_distr}.
It is in this sense that the strength laws express that $G$ is ``homomorphic''. We leave as an open problem to formulate precisely,
and prove, the full correspondence that strength laws are nothing but being ``homomorphic'' in the sense of our crafted monoidal target category.

\section{Conclusions and Further Work}\label{sec:conclusion}

We have presented a mathematical framework for the specification and generation of simply-typed syntax, fully implemented and computer-checked in \Coq, using the \UniMath library of univalent mathematics.
To this end, we generalized recent work by \citep{signatures_to_monads} on untyped syntax:
\begin{enumerate*}[label=(\roman*)]
\item we generalized definitions and adapted the statements and proofs---this required very little effort,
\item we extended the existing library on $\omega$-cocontinuous functors, and
\item we gave a new analysis of action-based strengths and embedded them into a bicategorical context.
\end{enumerate*}

It is difficult to provide exact numbers for what was involved in this
work as it is a modification and extension of a major library of
formalized mathematics. Indeed, some of the work was pure maintenance
of the library and is difficult to measure accurately, \eg, the notion
of isomorphism for monoidal categories had to be replaced which led to
a major overhaul of various files related to isomorphisms more
generally.
The tool \systemname{coqwc} counts approximately 6,700 lines of statements and proofs of new files, not counting the improvements and additions to existing material.

Our work relies heavily on \emph{universal properties}: both the syntax itself, as well as the substitution operation on it, are characterized through their universal properties.
Since universal properties are conditions of unique existence, with computational content,
the notions of contractible type and proposition provided by Univalent Foundations seem particularly well suited for discussing such work.

Our work benefited greatly from infrastructure provided by the \UniMath library.
Firstly, in \cref{sec:understanding} we heavily rely on the recent addition of
bicategory theory to \UniMath by \citet{DBLP:conf/rta/AhrensFMW19}.
A difficulty in bicategory
theory is the sheer complexity of the involved notions;
for instance a
bicategory in \cite{DBLP:conf/rta/AhrensFMW19} has $25$ fields of
increasing complexity.
By working in a proof assistant one can be
sure that no proof obligations have been overlooked.
This resembles
traditional metatheory of programming languages where many proofs
proceed by induction on large inductive definitions, leading to a multitude of
 cases to consider.
Secondly, the machinery of displayed categories~\citep{DBLP:journals/lmcs/AhrensL19} and displayed bicategories~\citep{DBLP:conf/rta/AhrensFMW19} helps building the complicated (bi)categories handled here.

A truly different approach to working with $\ccsort$ in the case of
$\CC = \set$ would be to work with the equivalent slice category
$\set / \sort$. The first version of the framework presented here was
in fact implemented this way. However, it had many drawbacks: first of
all it is less general, second, it was very clumsy to work with the
constructed monad on the slice category, third, working in a specific
category is often less convenient than in an abstract category (\eg,
things often unfold further, quickly leading to unreadable goals).

We have focused on the mathematical construction of abstract syntax, and on the mathematical justification of recursion on that syntax. The terms of our syntax associated to a signature are equivalence classes, and thus not convenient to handle in practice. A suitable inductive datatype of terms as used, \eg by \citet{DBLP:journals/jfrea/AhrensZ11} would be more convenient to use. We anticipate that it is feasible to construct, generically, this datatype and an isomorphism between our type family of terms and the inductive family.

It is not immediately clear how to directly generalize our approach to dependently-typed syntax (going beyond the pre-syntax of \cref{ex:coc}).
\citet{c-sys-from-relative-module} set out a path for the construction of such syntax;
it leads via pre-types and -terms in the form of monads as we construct them here,
and considers C-systems (\aka contextual categories \citep{DBLP:journals/apal/Cartmell86}) built from such pre-syntax as the raw material out of which to carve out the desired syntax.
The constructions presented here thus constitute an important stepping stone in Voevodsky's programme of formalizing the metatheory of type theory in Univalent Foundations.

\begin{acks}
We are grateful to Vladimir Voevodsky for many discussions on the subjects
of multi-sorted term systems and the intended applications to type theory.
We thank Peter LeFanu Lumsdaine for suggesting and formalizing an improvement
to one of our proofs.

Anders Mörtberg was supported by the Swedish Research Council (SRC,
Vetenskapsrådet) under Grant No.~2019-04545.

This material is based upon work supported by the National Science
Foundation under agreement No. DMS-1128155 and CMU 1150129-338510.

This work has
been partly funded by the CoqHoTT ERC Grant 637339.
\end{acks}

\bibliographystyle{ACM-Reference-Format}
\bibliography{literature}


\begin{thebibliography}{65}


\ifx \showCODEN    \undefined \def \showCODEN     #1{\unskip}     \fi
\ifx \showDOI      \undefined \def \showDOI       #1{#1}\fi
\ifx \showISBNx    \undefined \def \showISBNx     #1{\unskip}     \fi
\ifx \showISBNxiii \undefined \def \showISBNxiii  #1{\unskip}     \fi
\ifx \showISSN     \undefined \def \showISSN      #1{\unskip}     \fi
\ifx \showLCCN     \undefined \def \showLCCN      #1{\unskip}     \fi
\ifx \shownote     \undefined \def \shownote      #1{#1}          \fi
\ifx \showarticletitle \undefined \def \showarticletitle #1{#1}   \fi
\ifx \showURL      \undefined \def \showURL       {\relax}        \fi
\providecommand\bibfield[2]{#2}
\providecommand\bibinfo[2]{#2}
\providecommand\natexlab[1]{#1}
\providecommand\showeprint[2][]{arXiv:#2}

\bibitem[\protect\citeauthoryear{Ad\'{a}mek}{Ad\'{a}mek}{1974}]%
        {Adamek74}
\bibfield{author}{\bibinfo{person}{Ji\v{r}\'{i} Ad\'{a}mek}.}
  \bibinfo{year}{1974}\natexlab{}.
\newblock \showarticletitle{Free algebras and automata realizations in the
  language of categories}.
\newblock \bibinfo{journal}{\emph{Commentationes Mathematicae Universitatis
  Carolinae}} \bibinfo{volume}{15}, \bibinfo{number}{4} (\bibinfo{year}{1974}),
  \bibinfo{pages}{589--602}.
\newblock


\bibitem[\protect\citeauthoryear{Ahrens}{Ahrens}{2012a}]%
        {DBLP:journals/corr/abs-1107-4751}
\bibfield{author}{\bibinfo{person}{Benedikt Ahrens}.}
  \bibinfo{year}{2012}\natexlab{a}.
\newblock \showarticletitle{Extended Initiality for Typed Abstract Syntax}.
\newblock \bibinfo{journal}{\emph{Log. Methods Comput. Sci.}}
  \bibinfo{volume}{8}, \bibinfo{number}{2} (\bibinfo{year}{2012}).
\newblock
\urldef\tempurl%
\url{https://doi.org/10.2168/LMCS-8(2:1)2012}
\showDOI{\tempurl}


\bibitem[\protect\citeauthoryear{Ahrens}{Ahrens}{2012b}]%
        {DBLP:conf/wollic/Ahrens12}
\bibfield{author}{\bibinfo{person}{Benedikt Ahrens}.}
  \bibinfo{year}{2012}\natexlab{b}.
\newblock \showarticletitle{Initiality for Typed Syntax and Semantics}. In
  \bibinfo{booktitle}{\emph{Logic, Language, Information and Computation - 19th
  International Workshop, WoLLIC 2012, Buenos Aires, Argentina, September 3-6,
  2012. Proceedings}} \emph{(\bibinfo{series}{Lecture Notes in Computer
  Science}, Vol.~\bibinfo{volume}{7456})},
  \bibfield{editor}{\bibinfo{person}{C.{-}H.~Luke Ong} {and}
  \bibinfo{person}{Ruy J. G.~B. de~Queiroz}} (Eds.).
  \bibinfo{publisher}{Springer}, \bibinfo{pages}{127--141}.
\newblock
\urldef\tempurl%
\url{https://doi.org/10.1007/978-3-642-32621-9\_10}
\showDOI{\tempurl}


\bibitem[\protect\citeauthoryear{Ahrens}{Ahrens}{2016}]%
        {DBLP:journals/mscs/Ahrens16}
\bibfield{author}{\bibinfo{person}{Benedikt Ahrens}.}
  \bibinfo{year}{2016}\natexlab{}.
\newblock \showarticletitle{Modules over relative monads for syntax and
  semantics}.
\newblock \bibinfo{journal}{\emph{Math. Struct. Comput. Sci.}}
  \bibinfo{volume}{26}, \bibinfo{number}{1} (\bibinfo{year}{2016}),
  \bibinfo{pages}{3--37}.
\newblock
\urldef\tempurl%
\url{https://doi.org/10.1017/S0960129514000103}
\showDOI{\tempurl}


\bibitem[\protect\citeauthoryear{Ahrens}{Ahrens}{2019}]%
        {DBLP:journals/lmcs/Ahrens19}
\bibfield{author}{\bibinfo{person}{Benedikt Ahrens}.}
  \bibinfo{year}{2019}\natexlab{}.
\newblock \showarticletitle{Initial Semantics for Reduction Rules}.
\newblock \bibinfo{journal}{\emph{Log. Methods Comput. Sci.}}
  \bibinfo{volume}{15}, \bibinfo{number}{1} (\bibinfo{year}{2019}).
\newblock
\urldef\tempurl%
\url{https://doi.org/10.23638/LMCS-15(1:28)2019}
\showDOI{\tempurl}


\bibitem[\protect\citeauthoryear{Ahrens, Frumin, Maggesi, and van~der
  Weide}{Ahrens et~al\mbox{.}}{2019a}]%
        {DBLP:conf/rta/AhrensFMW19}
\bibfield{author}{\bibinfo{person}{Benedikt Ahrens}, \bibinfo{person}{Dan
  Frumin}, \bibinfo{person}{Marco Maggesi}, {and} \bibinfo{person}{Niels
  van~der Weide}.} \bibinfo{year}{2019}\natexlab{a}.
\newblock \showarticletitle{Bicategories in Univalent Foundations}. In
  \bibinfo{booktitle}{\emph{4th International Conference on Formal Structures
  for Computation and Deduction, {FSCD} 2019, June 24-30, 2019, Dortmund,
  Germany}} \emph{(\bibinfo{series}{LIPIcs}, Vol.~\bibinfo{volume}{131})},
  \bibfield{editor}{\bibinfo{person}{Herman Geuvers}} (Ed.).
  \bibinfo{publisher}{Schloss Dagstuhl - Leibniz-Zentrum f{\"{u}}r Informatik},
  \bibinfo{pages}{5:1--5:17}.
\newblock
\urldef\tempurl%
\url{https://doi.org/10.4230/LIPIcs.FSCD.2019.5}
\showDOI{\tempurl}


\bibitem[\protect\citeauthoryear{Ahrens, Hirschowitz, Lafont, and
  Maggesi}{Ahrens et~al\mbox{.}}{2019b}]%
        {DBLP:conf/rta/AhrensHLM19}
\bibfield{author}{\bibinfo{person}{Benedikt Ahrens},
  \bibinfo{person}{Andr{\'{e}} Hirschowitz}, \bibinfo{person}{Ambroise Lafont},
  {and} \bibinfo{person}{Marco Maggesi}.} \bibinfo{year}{2019}\natexlab{b}.
\newblock \showarticletitle{Modular Specification of Monads Through
  Higher-Order Presentations}. In \bibinfo{booktitle}{\emph{4th International
  Conference on Formal Structures for Computation and Deduction, {FSCD} 2019,
  June 24-30, 2019, Dortmund, Germany}} \emph{(\bibinfo{series}{LIPIcs},
  Vol.~\bibinfo{volume}{131})}, \bibfield{editor}{\bibinfo{person}{Herman
  Geuvers}} (Ed.). \bibinfo{publisher}{Schloss Dagstuhl - Leibniz-Zentrum
  f{\"{u}}r Informatik}, \bibinfo{pages}{6:1--6:19}.
\newblock
\urldef\tempurl%
\url{https://doi.org/10.4230/LIPIcs.FSCD.2019.6}
\showDOI{\tempurl}


\bibitem[\protect\citeauthoryear{Ahrens, Hirschowitz, Lafont, and
  Maggesi}{Ahrens et~al\mbox{.}}{2020}]%
        {DBLP:journals/pacmpl/AhrensHLM20}
\bibfield{author}{\bibinfo{person}{Benedikt Ahrens},
  \bibinfo{person}{Andr{\'{e}} Hirschowitz}, \bibinfo{person}{Ambroise Lafont},
  {and} \bibinfo{person}{Marco Maggesi}.} \bibinfo{year}{2020}\natexlab{}.
\newblock \showarticletitle{Reduction monads and their signatures}.
\newblock \bibinfo{journal}{\emph{Proc. {ACM} Program. Lang.}}
  \bibinfo{volume}{4}, \bibinfo{number}{{POPL}} (\bibinfo{year}{2020}),
  \bibinfo{pages}{31:1--31:29}.
\newblock
\urldef\tempurl%
\url{https://doi.org/10.1145/3371099}
\showDOI{\tempurl}


\bibitem[\protect\citeauthoryear{Ahrens, Hirschowitz, Lafont, and
  Maggesi}{Ahrens et~al\mbox{.}}{2021}]%
        {DBLP:journals/lmcs/AhrensHLM21}
\bibfield{author}{\bibinfo{person}{Benedikt Ahrens},
  \bibinfo{person}{Andr{\'{e}} Hirschowitz}, \bibinfo{person}{Ambroise Lafont},
  {and} \bibinfo{person}{Marco Maggesi}.} \bibinfo{year}{2021}\natexlab{}.
\newblock \showarticletitle{Presentable signatures and initial semantics}.
\newblock \bibinfo{journal}{\emph{Log. Methods Comput. Sci.}}
  \bibinfo{volume}{17}, \bibinfo{number}{2} (\bibinfo{year}{2021}).
\newblock
\urldef\tempurl%
\url{https://doi.org/10.23638/LMCS-17(2:17)2021}
\showDOI{\tempurl}


\bibitem[\protect\citeauthoryear{Ahrens, Kapulkin, and Shulman}{Ahrens
  et~al\mbox{.}}{2015}]%
        {rezk_completion}
\bibfield{author}{\bibinfo{person}{Benedikt Ahrens}, \bibinfo{person}{Krzysztof
  Kapulkin}, {and} \bibinfo{person}{Michael Shulman}.}
  \bibinfo{year}{2015}\natexlab{}.
\newblock \showarticletitle{{Univalent categories and the Rezk completion}}.
\newblock \bibinfo{journal}{\emph{Math.\ Struct.\ in Comp.\ Science}}
  \bibinfo{volume}{25} (\bibinfo{year}{2015}), \bibinfo{pages}{1010--1039}.
\newblock
\urldef\tempurl%
\url{https://doi.org/10.1017/S0960129514000486}
\showDOI{\tempurl}
\showeprint[arxiv]{1303.0584}


\bibitem[\protect\citeauthoryear{Ahrens and Lumsdaine}{Ahrens and
  Lumsdaine}{2019}]%
        {DBLP:journals/lmcs/AhrensL19}
\bibfield{author}{\bibinfo{person}{Benedikt Ahrens} {and}
  \bibinfo{person}{Peter~LeFanu Lumsdaine}.} \bibinfo{year}{2019}\natexlab{}.
\newblock \showarticletitle{Displayed Categories}.
\newblock \bibinfo{journal}{\emph{Log. Methods Comput. Sci.}}
  \bibinfo{volume}{15}, \bibinfo{number}{1} (\bibinfo{year}{2019}).
\newblock
\urldef\tempurl%
\url{https://doi.org/10.23638/LMCS-15(1:20)2019}
\showDOI{\tempurl}


\bibitem[\protect\citeauthoryear{Ahrens and Matthes}{Ahrens and
  Matthes}{2015}]%
        {DBLP:conf/types/AhrensM15}
\bibfield{author}{\bibinfo{person}{Benedikt Ahrens} {and}
  \bibinfo{person}{Ralph Matthes}.} \bibinfo{year}{2015}\natexlab{}.
\newblock \showarticletitle{Heterogeneous Substitution Systems Revisited}. In
  \bibinfo{booktitle}{\emph{21st International Conference on Types for Proofs
  and Programs, {TYPES} 2015, May 18-21, 2015, Tallinn, Estonia}}
  \emph{(\bibinfo{series}{LIPIcs}, Vol.~\bibinfo{volume}{69})},
  \bibfield{editor}{\bibinfo{person}{Tarmo Uustalu}} (Ed.).
  \bibinfo{publisher}{Schloss Dagstuhl - Leibniz-Zentrum f{\"{u}}r Informatik},
  \bibinfo{pages}{2:1--2:23}.
\newblock
\urldef\tempurl%
\url{https://doi.org/10.4230/LIPIcs.TYPES.2015.2}
\showDOI{\tempurl}


\bibitem[\protect\citeauthoryear{Ahrens, Matthes, and M{\"{o}}rtberg}{Ahrens
  et~al\mbox{.}}{2019c}]%
        {signatures_to_monads}
\bibfield{author}{\bibinfo{person}{Benedikt Ahrens}, \bibinfo{person}{Ralph
  Matthes}, {and} \bibinfo{person}{Anders M{\"{o}}rtberg}.}
  \bibinfo{year}{2019}\natexlab{c}.
\newblock \showarticletitle{From Signatures to Monads in {UniMath}}.
\newblock \bibinfo{journal}{\emph{J. Autom. Reason.}} \bibinfo{volume}{63},
  \bibinfo{number}{2} (\bibinfo{year}{2019}), \bibinfo{pages}{285--318}.
\newblock
\urldef\tempurl%
\url{https://doi.org/10.1007/s10817-018-9474-4}
\showDOI{\tempurl}


\bibitem[\protect\citeauthoryear{Ahrens and Zsidó}{Ahrens and Zsidó}{2011}]%
        {DBLP:journals/jfrea/AhrensZ11}
\bibfield{author}{\bibinfo{person}{Benedikt Ahrens} {and}
  \bibinfo{person}{Julianna Zsidó}.} \bibinfo{year}{2011}\natexlab{}.
\newblock \showarticletitle{Initial Semantics for higher-order typed syntax in
  Coq}.
\newblock \bibinfo{journal}{\emph{J. Formaliz. Reason.}} \bibinfo{volume}{4},
  \bibinfo{number}{1} (\bibinfo{year}{2011}), \bibinfo{pages}{25--69}.
\newblock
\urldef\tempurl%
\url{https://doi.org/10.6092/issn.1972-5787/2066}
\showDOI{\tempurl}


\bibitem[\protect\citeauthoryear{Allais, Atkey, Chapman, McBride, and
  McKinna}{Allais et~al\mbox{.}}{2018}]%
        {DBLP:journals/pacmpl/AllaisA0MM18}
\bibfield{author}{\bibinfo{person}{Guillaume Allais}, \bibinfo{person}{Robert
  Atkey}, \bibinfo{person}{James Chapman}, \bibinfo{person}{Conor McBride},
  {and} \bibinfo{person}{James McKinna}.} \bibinfo{year}{2018}\natexlab{}.
\newblock \showarticletitle{A type and scope safe universe of syntaxes with
  binding: their semantics and proofs}.
\newblock \bibinfo{journal}{\emph{Proc. {ACM} Program. Lang.}}
  \bibinfo{volume}{2}, \bibinfo{number}{{ICFP}} (\bibinfo{year}{2018}),
  \bibinfo{pages}{90:1--90:30}.
\newblock
\urldef\tempurl%
\url{https://doi.org/10.1145/3236785}
\showDOI{\tempurl}


\bibitem[\protect\citeauthoryear{Allais, Atkey, Chapman, McBride, and
  McKinna}{Allais et~al\mbox{.}}{2020}]%
        {DBLP:journals/corr/abs-2001-11001}
\bibfield{author}{\bibinfo{person}{Guillaume Allais}, \bibinfo{person}{Robert
  Atkey}, \bibinfo{person}{James Chapman}, \bibinfo{person}{Conor McBride},
  {and} \bibinfo{person}{James McKinna}.} \bibinfo{year}{2020}\natexlab{}.
\newblock \showarticletitle{A Type and Scope Safe Universe of Syntaxes with
  Binding: Their Semantics and Proofs}.
\newblock \bibinfo{journal}{\emph{CoRR}}  \bibinfo{volume}{abs/2001.11001}
  (\bibinfo{year}{2020}).
\newblock
\showeprint[arXiv]{2001.11001}
\urldef\tempurl%
\url{https://arxiv.org/abs/2001.11001}
\showURL{%
\tempurl}


\bibitem[\protect\citeauthoryear{Altenkirch, Ghani, Hancock, McBride, and
  Morris}{Altenkirch et~al\mbox{.}}{2015}]%
        {DBLP:journals/jfp/AltenkirchGHMM15}
\bibfield{author}{\bibinfo{person}{Thorsten Altenkirch}, \bibinfo{person}{Neil
  Ghani}, \bibinfo{person}{Peter~G. Hancock}, \bibinfo{person}{Conor McBride},
  {and} \bibinfo{person}{Peter Morris}.} \bibinfo{year}{2015}\natexlab{}.
\newblock \showarticletitle{Indexed containers}.
\newblock \bibinfo{journal}{\emph{J. Funct. Program.}}  \bibinfo{volume}{25}
  (\bibinfo{year}{2015}).
\newblock
\urldef\tempurl%
\url{https://doi.org/10.1017/S095679681500009X}
\showDOI{\tempurl}


\bibitem[\protect\citeauthoryear{Altenkirch and Reus}{Altenkirch and
  Reus}{1999}]%
        {DBLP:conf/csl/AltenkirchR99}
\bibfield{author}{\bibinfo{person}{Thorsten Altenkirch} {and}
  \bibinfo{person}{Bernhard Reus}.} \bibinfo{year}{1999}\natexlab{}.
\newblock \showarticletitle{Monadic Presentations of Lambda Terms Using
  Generalized Inductive Types}. In \bibinfo{booktitle}{\emph{Computer Science
  Logic, 13th International Workshop, {CSL} '99, 8th Annual Conference of the
  EACSL, Madrid, Spain, September 20-25, 1999, Proceedings}}
  \emph{(\bibinfo{series}{Lecture Notes in Computer Science},
  Vol.~\bibinfo{volume}{1683})}, \bibfield{editor}{\bibinfo{person}{J{\"{o}}rg
  Flum} {and} \bibinfo{person}{Mario Rodr{\'{\i}}guez{-}Artalejo}} (Eds.).
  \bibinfo{publisher}{Springer}, \bibinfo{pages}{453--468}.
\newblock
\urldef\tempurl%
\url{https://doi.org/10.1007/3-540-48168-0\_32}
\showDOI{\tempurl}


\bibitem[\protect\citeauthoryear{Bellegarde and Hook}{Bellegarde and
  Hook}{1994}]%
        {DBLP:journals/scp/BellegardeH94}
\bibfield{author}{\bibinfo{person}{Fran{\c{c}}oise Bellegarde} {and}
  \bibinfo{person}{James Hook}.} \bibinfo{year}{1994}\natexlab{}.
\newblock \showarticletitle{Substitution: {A} Formal Methods Case Study Using
  Monads and Transformations}.
\newblock \bibinfo{journal}{\emph{Sci. Comput. Program.}} \bibinfo{volume}{23},
  \bibinfo{number}{2-3} (\bibinfo{year}{1994}), \bibinfo{pages}{287--311}.
\newblock
\urldef\tempurl%
\url{https://doi.org/10.1016/0167-6423(94)00022-0}
\showDOI{\tempurl}


\bibitem[\protect\citeauthoryear{Bird and Paterson}{Bird and Paterson}{1999}]%
        {DBLP:journals/jfp/BirdP99}
\bibfield{author}{\bibinfo{person}{Richard~S. Bird} {and} \bibinfo{person}{Ross
  Paterson}.} \bibinfo{year}{1999}\natexlab{}.
\newblock \showarticletitle{De Bruijn Notation as a Nested Datatype}.
\newblock \bibinfo{journal}{\emph{J. Funct. Program.}} \bibinfo{volume}{9},
  \bibinfo{number}{1} (\bibinfo{year}{1999}), \bibinfo{pages}{77--91}.
\newblock
\urldef\tempurl%
\url{http://journals.cambridge.org/action/displayAbstract?aid=44239}
\showURL{%
\tempurl}


\bibitem[\protect\citeauthoryear{Cartmell}{Cartmell}{1986}]%
        {DBLP:journals/apal/Cartmell86}
\bibfield{author}{\bibinfo{person}{John Cartmell}.}
  \bibinfo{year}{1986}\natexlab{}.
\newblock \showarticletitle{Generalised algebraic theories and contextual
  categories}.
\newblock \bibinfo{journal}{\emph{Ann. Pure Appl. Logic}}  \bibinfo{volume}{32}
  (\bibinfo{year}{1986}), \bibinfo{pages}{209--243}.
\newblock
\urldef\tempurl%
\url{https://doi.org/10.1016/0168-0072(86)90053-9}
\showDOI{\tempurl}


\bibitem[\protect\citeauthoryear{Chapman, Dagand, McBride, and Morris}{Chapman
  et~al\mbox{.}}{2010}]%
        {DBLP:conf/icfp/ChapmanDMM10}
\bibfield{author}{\bibinfo{person}{James Chapman},
  \bibinfo{person}{Pierre{-}{\'{E}}variste Dagand}, \bibinfo{person}{Conor
  McBride}, {and} \bibinfo{person}{Peter Morris}.}
  \bibinfo{year}{2010}\natexlab{}.
\newblock \showarticletitle{The gentle art of levitation}. In
  \bibinfo{booktitle}{\emph{Proceeding of the 15th {ACM} {SIGPLAN}
  international conference on Functional programming, {ICFP} 2010, Baltimore,
  Maryland, USA, September 27-29, 2010}},
  \bibfield{editor}{\bibinfo{person}{Paul Hudak} {and}
  \bibinfo{person}{Stephanie Weirich}} (Eds.). \bibinfo{publisher}{{ACM}},
  \bibinfo{pages}{3--14}.
\newblock
\urldef\tempurl%
\url{https://doi.org/10.1145/1863543.1863547}
\showDOI{\tempurl}


\bibitem[\protect\citeauthoryear{Dagand and McBride}{Dagand and
  McBride}{2013}]%
        {DBLP:conf/lics/DagandM13}
\bibfield{author}{\bibinfo{person}{Pierre{-}{\'{E}}variste Dagand} {and}
  \bibinfo{person}{Conor McBride}.} \bibinfo{year}{2013}\natexlab{}.
\newblock \showarticletitle{A Categorical Treatment of Ornaments}. In
  \bibinfo{booktitle}{\emph{28th Annual {ACM/IEEE} Symposium on Logic in
  Computer Science, {LICS} 2013, New Orleans, LA, USA, June 25-28, 2013}}.
  \bibinfo{publisher}{{IEEE} Computer Society}, \bibinfo{pages}{530--539}.
\newblock
\urldef\tempurl%
\url{https://doi.org/10.1109/LICS.2013.60}
\showDOI{\tempurl}


\bibitem[\protect\citeauthoryear{Fiore and Szamozvancev}{Fiore and
  Szamozvancev}{2021}]%
        {fiore-szamozvancev}
\bibfield{author}{\bibinfo{person}{Marcelo Fiore} {and}
  \bibinfo{person}{Dmitrij Szamozvancev}.} \bibinfo{year}{2021}\natexlab{}.
\newblock \showarticletitle{Formal Metatheory of Second-Order Abstract Syntax}.
\newblock \bibinfo{journal}{\emph{To be published in Principles of Programming
  Languages (POPL) 2022}} (\bibinfo{year}{2021}).
\newblock
\urldef\tempurl%
\url{https://www.repository.cam.ac.uk/handle/1810/330658}
\showURL{%
\tempurl}


\bibitem[\protect\citeauthoryear{Fiore}{Fiore}{2002}]%
        {DBLP:conf/ppdp/Fiore02}
\bibfield{author}{\bibinfo{person}{Marcelo~P. Fiore}.}
  \bibinfo{year}{2002}\natexlab{}.
\newblock \showarticletitle{Semantic analysis of normalisation by evaluation
  for typed lambda calculus}. In \bibinfo{booktitle}{\emph{Proceedings of the
  4th international {ACM} {SIGPLAN} conference on Principles and practice of
  declarative programming, October 6-8, 2002, Pittsburgh, PA, {USA} (Affiliated
  with {PLI} 2002)}}. \bibinfo{publisher}{{ACM}}, \bibinfo{pages}{26--37}.
\newblock
\urldef\tempurl%
\url{https://doi.org/10.1145/571157.571161}
\showDOI{\tempurl}


\bibitem[\protect\citeauthoryear{Fiore}{Fiore}{2008}]%
        {DBLP:conf/lics/Fiore08}
\bibfield{author}{\bibinfo{person}{Marcelo~P. Fiore}.}
  \bibinfo{year}{2008}\natexlab{}.
\newblock \showarticletitle{Second-Order and Dependently-Sorted Abstract
  Syntax}. In \bibinfo{booktitle}{\emph{Proceedings of the Twenty-Third Annual
  {IEEE} Symposium on Logic in Computer Science, {LICS} 2008, 24-27 June 2008,
  Pittsburgh, PA, {USA}}}. \bibinfo{publisher}{{IEEE} Computer Society},
  \bibinfo{pages}{57--68}.
\newblock
\showISBNx{978-0-7695-3183-0}
\urldef\tempurl%
\url{https://doi.org/10.1109/LICS.2008.38}
\showDOI{\tempurl}


\bibitem[\protect\citeauthoryear{Fiore}{Fiore}{2012}]%
        {DBLP:conf/icalp/Fiore12}
\bibfield{author}{\bibinfo{person}{Marcelo~P. Fiore}.}
  \bibinfo{year}{2012}\natexlab{}.
\newblock \showarticletitle{Discrete Generalised Polynomial Functors -
  (Extended Abstract)}. In \bibinfo{booktitle}{\emph{Automata, Languages, and
  Programming - 39th International Colloquium, {ICALP} 2012, Warwick, UK, July
  9-13, 2012, Proceedings, Part {II}}} \emph{(\bibinfo{series}{Lecture Notes in
  Computer Science}, Vol.~\bibinfo{volume}{7392})},
  \bibfield{editor}{\bibinfo{person}{Artur Czumaj}, \bibinfo{person}{Kurt
  Mehlhorn}, \bibinfo{person}{Andrew~M. Pitts}, {and} \bibinfo{person}{Roger
  Wattenhofer}} (Eds.). \bibinfo{publisher}{Springer},
  \bibinfo{pages}{214--226}.
\newblock
\urldef\tempurl%
\url{https://doi.org/10.1007/978-3-642-31585-5\_22}
\showDOI{\tempurl}


\bibitem[\protect\citeauthoryear{Fiore and Hamana}{Fiore and Hamana}{2013}]%
        {DBLP:conf/lics/FioreH13}
\bibfield{author}{\bibinfo{person}{Marcelo~P. Fiore} {and}
  \bibinfo{person}{Makoto Hamana}.} \bibinfo{year}{2013}\natexlab{}.
\newblock \showarticletitle{Multiversal Polymorphic Algebraic Theories: Syntax,
  Semantics, Translations, and Equational Logic}. In
  \bibinfo{booktitle}{\emph{28th Annual {ACM/IEEE} Symposium on Logic in
  Computer Science, {LICS} 2013, New Orleans, LA, USA, June 25-28, 2013}}.
  \bibinfo{publisher}{{IEEE} Computer Society}, \bibinfo{pages}{520--529}.
\newblock
\urldef\tempurl%
\url{https://doi.org/10.1109/LICS.2013.59}
\showDOI{\tempurl}


\bibitem[\protect\citeauthoryear{Fiore, Plotkin, and Turi}{Fiore
  et~al\mbox{.}}{1999}]%
        {DBLP:conf/lics/FiorePT99}
\bibfield{author}{\bibinfo{person}{Marcelo~P. Fiore},
  \bibinfo{person}{Gordon~D. Plotkin}, {and} \bibinfo{person}{Daniele Turi}.}
  \bibinfo{year}{1999}\natexlab{}.
\newblock \showarticletitle{Abstract Syntax and Variable Binding}. In
  \bibinfo{booktitle}{\emph{14th Annual {IEEE} Symposium on Logic in Computer
  Science, Trento, Italy, July 2-5, 1999}}. \bibinfo{publisher}{{IEEE} Computer
  Society}, \bibinfo{pages}{193--202}.
\newblock
\urldef\tempurl%
\url{https://doi.org/10.1109/LICS.1999.782615}
\showDOI{\tempurl}


\bibitem[\protect\citeauthoryear{Gabbay and Pitts}{Gabbay and Pitts}{1999}]%
        {DBLP:conf/lics/GabbayP99}
\bibfield{author}{\bibinfo{person}{Murdoch Gabbay} {and}
  \bibinfo{person}{Andrew~M. Pitts}.} \bibinfo{year}{1999}\natexlab{}.
\newblock \showarticletitle{A New Approach to Abstract Syntax Involving
  Binders}. In \bibinfo{booktitle}{\emph{14th Annual {IEEE} Symposium on Logic
  in Computer Science, Trento, Italy, July 2-5, 1999}}.
  \bibinfo{publisher}{{IEEE} Computer Society}, \bibinfo{pages}{214--224}.
\newblock
\urldef\tempurl%
\url{https://doi.org/10.1109/LICS.1999.782617}
\showDOI{\tempurl}


\bibitem[\protect\citeauthoryear{Gambino and Hyland}{Gambino and
  Hyland}{2003}]%
        {DBLP:conf/types/GambinoH03}
\bibfield{author}{\bibinfo{person}{Nicola Gambino} {and}
  \bibinfo{person}{Martin Hyland}.} \bibinfo{year}{2003}\natexlab{}.
\newblock \showarticletitle{Wellfounded Trees and Dependent Polynomial
  Functors}. In \bibinfo{booktitle}{\emph{Types for Proofs and Programs,
  International Workshop, {TYPES} 2003, Torino, Italy, April 30 - May 4, 2003,
  Revised Selected Papers}} \emph{(\bibinfo{series}{Lecture Notes in Computer
  Science}, Vol.~\bibinfo{volume}{3085})},
  \bibfield{editor}{\bibinfo{person}{Stefano Berardi}, \bibinfo{person}{Mario
  Coppo}, {and} \bibinfo{person}{Ferruccio Damiani}} (Eds.).
  \bibinfo{publisher}{Springer}, \bibinfo{pages}{210--225}.
\newblock
\urldef\tempurl%
\url{https://doi.org/10.1007/978-3-540-24849-1\_14}
\showDOI{\tempurl}


\bibitem[\protect\citeauthoryear{Hamana}{Hamana}{2011}]%
        {DBLP:conf/fossacs/Hamana11}
\bibfield{author}{\bibinfo{person}{Makoto Hamana}.}
  \bibinfo{year}{2011}\natexlab{}.
\newblock \showarticletitle{Polymorphic Abstract Syntax via Grothendieck
  Construction}. In \bibinfo{booktitle}{\emph{Foundations of Software Science
  and Computational Structures - 14th International Conference, {FOSSACS} 2011,
  Held as Part of the Joint European Conferences on Theory and Practice of
  Software, {ETAPS} 2011, Saarbr{\"{u}}cken, Germany, March 26-April 3, 2011.
  Proceedings}} \emph{(\bibinfo{series}{Lecture Notes in Computer Science},
  Vol.~\bibinfo{volume}{6604})}, \bibfield{editor}{\bibinfo{person}{Martin
  Hofmann}} (Ed.). \bibinfo{publisher}{Springer}, \bibinfo{pages}{381--395}.
\newblock
\urldef\tempurl%
\url{https://doi.org/10.1007/978-3-642-19805-2\_26}
\showDOI{\tempurl}


\bibitem[\protect\citeauthoryear{Hamana and Fiore}{Hamana and Fiore}{2011}]%
        {DBLP:conf/icfp/HamanaF11}
\bibfield{author}{\bibinfo{person}{Makoto Hamana} {and}
  \bibinfo{person}{Marcelo~P. Fiore}.} \bibinfo{year}{2011}\natexlab{}.
\newblock \showarticletitle{A foundation for GADTs and inductive families:
  dependent polynomial functor approach}. In
  \bibinfo{booktitle}{\emph{Proceedings of the seventh {ACM} {SIGPLAN} workshop
  on Generic programming, WGP@ICFP 2011, Tokyo, Japan, September 19-21, 2011}},
  \bibfield{editor}{\bibinfo{person}{Jaakko J{\"{a}}rvi} {and}
  \bibinfo{person}{Shin{-}Cheng Mu}} (Eds.). \bibinfo{publisher}{{ACM}},
  \bibinfo{pages}{59--70}.
\newblock
\urldef\tempurl%
\url{https://doi.org/10.1145/2036918.2036927}
\showDOI{\tempurl}


\bibitem[\protect\citeauthoryear{Hirschowitz, Hirschowitz, and
  Lafont}{Hirschowitz et~al\mbox{.}}{2020a}]%
        {DBLP:conf/fscd/HirschowitzHL20}
\bibfield{author}{\bibinfo{person}{Andr{\'{e}} Hirschowitz},
  \bibinfo{person}{Tom Hirschowitz}, {and} \bibinfo{person}{Ambroise Lafont}.}
  \bibinfo{year}{2020}\natexlab{a}.
\newblock \showarticletitle{Modules over Monads and Operational Semantics}. In
  \bibinfo{booktitle}{\emph{5th International Conference on Formal Structures
  for Computation and Deduction, {FSCD} 2020, June 29-July 6, 2020, Paris,
  France (Virtual Conference)}} \emph{(\bibinfo{series}{LIPIcs},
  Vol.~\bibinfo{volume}{167})}, \bibfield{editor}{\bibinfo{person}{Zena~M.
  Ariola}} (Ed.). \bibinfo{publisher}{Schloss Dagstuhl - Leibniz-Zentrum
  f{\"{u}}r Informatik}, \bibinfo{pages}{12:1--12:23}.
\newblock
\urldef\tempurl%
\url{https://doi.org/10.4230/LIPIcs.FSCD.2020.12}
\showDOI{\tempurl}


\bibitem[\protect\citeauthoryear{Hirschowitz, Hirschowitz, and
  Lafont}{Hirschowitz et~al\mbox{.}}{2020b}]%
        {HHL20-ext}
\bibfield{author}{\bibinfo{person}{Andr{\'{e}} Hirschowitz},
  \bibinfo{person}{Tom Hirschowitz}, {and} \bibinfo{person}{Ambroise Lafont}.}
  \bibinfo{year}{2020}\natexlab{b}.
\newblock \showarticletitle{Modules over Monads and Operational Semantics}.
\newblock  (\bibinfo{year}{2020}).
\newblock
\showeprint{2012.06530v1}


\bibitem[\protect\citeauthoryear{Hirschowitz and Maggesi}{Hirschowitz and
  Maggesi}{2007}]%
        {DBLP:conf/wollic/HirschowitzM07}
\bibfield{author}{\bibinfo{person}{Andr\'e Hirschowitz} {and}
  \bibinfo{person}{Marco Maggesi}.} \bibinfo{year}{2007}\natexlab{}.
\newblock \showarticletitle{{Modules over Monads and Linearity}}. In
  \bibinfo{booktitle}{\emph{WoLLIC}} \emph{(\bibinfo{series}{Lecture Notes in
  Computer Science}, Vol.~\bibinfo{volume}{4576})},
  \bibfield{editor}{\bibinfo{person}{Daniel Leivant} {and} \bibinfo{person}{Ruy
  J. G.~B. de~Queiroz}} (Eds.). \bibinfo{publisher}{Springer},
  \bibinfo{pages}{218--237}.
\newblock
\showISBNx{978-3-540-73443-7}


\bibitem[\protect\citeauthoryear{Hirschowitz and Maggesi}{Hirschowitz and
  Maggesi}{2010}]%
        {DBLP:journals/iandc/HirschowitzM10}
\bibfield{author}{\bibinfo{person}{Andr{\'e} Hirschowitz} {and}
  \bibinfo{person}{Marco Maggesi}.} \bibinfo{year}{2010}\natexlab{}.
\newblock \showarticletitle{Modules over monads and initial semantics}.
\newblock \bibinfo{journal}{\emph{Inf. Comput.}} \bibinfo{volume}{208},
  \bibinfo{number}{5} (\bibinfo{year}{2010}), \bibinfo{pages}{545--564}.
\newblock


\bibitem[\protect\citeauthoryear{Hofmann}{Hofmann}{1999}]%
        {DBLP:conf/lics/Hofmann99}
\bibfield{author}{\bibinfo{person}{Martin Hofmann}.}
  \bibinfo{year}{1999}\natexlab{}.
\newblock \showarticletitle{Semantical Analysis of Higher-Order Abstract
  Syntax}. In \bibinfo{booktitle}{\emph{14th Annual {IEEE} Symposium on Logic
  in Computer Science, Trento, Italy, July 2-5, 1999}}.
  \bibinfo{publisher}{{IEEE} Computer Society}, \bibinfo{pages}{204--213}.
\newblock
\urldef\tempurl%
\url{https://doi.org/10.1109/LICS.1999.782616}
\showDOI{\tempurl}


\bibitem[\protect\citeauthoryear{Hur}{Hur}{2010}]%
        {DBLP:phd/ethos/Hur10}
\bibfield{author}{\bibinfo{person}{Chung{-}Kil Hur}.}
  \bibinfo{year}{2010}\natexlab{}.
\newblock \emph{\bibinfo{title}{Categorical equational systems : algebraic
  models and equational reasoning}}.
\newblock \bibinfo{thesistype}{Ph.\,D. Dissertation}.
  \bibinfo{school}{University of Cambridge, {UK}}.
\newblock
\urldef\tempurl%
\url{http://ethos.bl.uk/OrderDetails.do?uin=uk.bl.ethos.608664}
\showURL{%
\tempurl}


\bibitem[\protect\citeauthoryear{Janelidze and Kelly}{Janelidze and
  Kelly}{2001}]%
        {JanelidzeKelly2001}
\bibfield{author}{\bibinfo{person}{George Janelidze} {and}
  \bibinfo{person}{Gregory~Maxwell Kelly}.} \bibinfo{year}{2001}\natexlab{}.
\newblock \showarticletitle{A Note on Actions of a Monoidal Category}.
\newblock \bibinfo{journal}{\emph{Theory and Applications of Categories}}
  \bibinfo{volume}{9}, \bibinfo{number}{4} (\bibinfo{year}{2001}),
  \bibinfo{pages}{61--91}.
\newblock
\newblock
\shownote{Online available at
  \url{http://www.tac.mta.ca/tac/volumes/9/n4/9-04abs.html}}.


\bibitem[\protect\citeauthoryear{Lee, d.~S.~Oliveira, Cho, and Yi}{Lee
  et~al\mbox{.}}{2012}]%
        {DBLP:conf/esop/LeeOCY12}
\bibfield{author}{\bibinfo{person}{Gyesik Lee}, \bibinfo{person}{Bruno~C. d.
  S.~Oliveira}, \bibinfo{person}{Sungkeun Cho}, {and}
  \bibinfo{person}{Kwangkeun Yi}.} \bibinfo{year}{2012}\natexlab{}.
\newblock \showarticletitle{GMeta: {A} Generic Formal Metatheory Framework for
  First-Order Representations}. In \bibinfo{booktitle}{\emph{Programming
  Languages and Systems - 21st European Symposium on Programming, {ESOP} 2012,
  Held as Part of the European Joint Conferences on Theory and Practice of
  Software, {ETAPS} 2012, Tallinn, Estonia, March 24 - April 1, 2012.
  Proceedings}} \emph{(\bibinfo{series}{Lecture Notes in Computer Science},
  Vol.~\bibinfo{volume}{7211})}, \bibfield{editor}{\bibinfo{person}{Helmut
  Seidl}} (Ed.). \bibinfo{publisher}{Springer}, \bibinfo{pages}{436--455}.
\newblock
\urldef\tempurl%
\url{https://doi.org/10.1007/978-3-642-28869-2\_22}
\showDOI{\tempurl}


\bibitem[\protect\citeauthoryear{L{\"{o}}h and Magalh{\~{a}}es}{L{\"{o}}h and
  Magalh{\~{a}}es}{2011}]%
        {DBLP:conf/icfp/LohM11}
\bibfield{author}{\bibinfo{person}{Andres L{\"{o}}h} {and}
  \bibinfo{person}{Jos{\'{e}}~Pedro Magalh{\~{a}}es}.}
  \bibinfo{year}{2011}\natexlab{}.
\newblock \showarticletitle{Generic programming with indexed functors}. In
  \bibinfo{booktitle}{\emph{Proceedings of the seventh {ACM} {SIGPLAN} workshop
  on Generic programming, WGP@ICFP 2011, Tokyo, Japan, September 19-21, 2011}},
  \bibfield{editor}{\bibinfo{person}{Jaakko J{\"{a}}rvi} {and}
  \bibinfo{person}{Shin{-}Cheng Mu}} (Eds.). \bibinfo{publisher}{{ACM}},
  \bibinfo{pages}{1--12}.
\newblock
\urldef\tempurl%
\url{https://doi.org/10.1145/2036918.2036920}
\showDOI{\tempurl}


\bibitem[\protect\citeauthoryear{Mac~Lane}{Mac~Lane}{1998}]%
        {maclane}
\bibfield{author}{\bibinfo{person}{Saunders Mac~Lane}.}
  \bibinfo{year}{1998}\natexlab{}.
\newblock \bibinfo{booktitle}{\emph{{Categories for the Working Mathematician}}
  (\bibinfo{edition}{second} ed.)}. \bibinfo{series}{Graduate Texts in
  Mathematics}, Vol.~\bibinfo{volume}{5}.
\newblock \bibinfo{publisher}{Springer-Verlag}, \bibinfo{address}{New York}.
  xii+314 pages.
\newblock
\showISBNx{0-387-98403-8}


\bibitem[\protect\citeauthoryear{Mahmoud}{Mahmoud}{2011}]%
        {DBLP:phd/ethos/Mahmoud11}
\bibfield{author}{\bibinfo{person}{Ola Mahmoud}.}
  \bibinfo{year}{2011}\natexlab{}.
\newblock \emph{\bibinfo{title}{Second-order algebraic theories}}.
\newblock \bibinfo{thesistype}{Ph.\,D. Dissertation}.
  \bibinfo{school}{University of Cambridge, {UK}}.
\newblock
\urldef\tempurl%
\url{http://www.dspace.cam.ac.uk/handle/1810/241035}
\showURL{%
\tempurl}


\bibitem[\protect\citeauthoryear{Martin{-}L{\"{o}}f}{Martin{-}L{\"{o}}f}{1984}]%
        {MLTT}
\bibfield{author}{\bibinfo{person}{Per Martin{-}L{\"{o}}f}.}
  \bibinfo{year}{1984}\natexlab{}.
\newblock \bibinfo{booktitle}{\emph{Intuitionistic type theory}}.
  \bibinfo{series}{Studies in proof theory}, Vol.~\bibinfo{volume}{1}.
\newblock \bibinfo{publisher}{Bibliopolis}.
\newblock
\showISBNx{978-88-7088-228-5}


\bibitem[\protect\citeauthoryear{Matthes and Uustalu}{Matthes and
  Uustalu}{2004}]%
        {DBLP:journals/tcs/MatthesU04}
\bibfield{author}{\bibinfo{person}{Ralph Matthes} {and} \bibinfo{person}{Tarmo
  Uustalu}.} \bibinfo{year}{2004}\natexlab{}.
\newblock \showarticletitle{Substitution in non-wellfounded syntax with
  variable binding}.
\newblock \bibinfo{journal}{\emph{Theoretical Computer Science}}
  \bibinfo{volume}{327}, \bibinfo{number}{1-2} (\bibinfo{year}{2004}),
  \bibinfo{pages}{155--174}.
\newblock
\urldef\tempurl%
\url{https://doi.org/10.1016/j.tcs.2004.07.025}
\showDOI{\tempurl}


\bibitem[\protect\citeauthoryear{Mendler}{Mendler}{1991}]%
        {DBLP:journals/apal/Mendler91}
\bibfield{author}{\bibinfo{person}{Nax~Paul Mendler}.}
  \bibinfo{year}{1991}\natexlab{}.
\newblock \showarticletitle{Inductive types and type constraints in the
  second-order lambda calculus}.
\newblock \bibinfo{journal}{\emph{Ann. Pure Appl. Logic}} \bibinfo{volume}{51},
  \bibinfo{number}{1-2} (\bibinfo{year}{1991}), \bibinfo{pages}{159--172}.
\newblock
\urldef\tempurl%
\url{https://doi.org/10.1016/0168-0072(91)90069-X}
\showDOI{\tempurl}


\bibitem[\protect\citeauthoryear{Miculan and Scagnetto}{Miculan and
  Scagnetto}{2003}]%
        {DBLP:conf/ppdp/MiculanS03}
\bibfield{author}{\bibinfo{person}{Marino Miculan} {and} \bibinfo{person}{Ivan
  Scagnetto}.} \bibinfo{year}{2003}\natexlab{}.
\newblock \showarticletitle{A framework for typed {HOAS} and semantics}. In
  \bibinfo{booktitle}{\emph{Proceedings of the 5th International {ACM}
  {SIGPLAN} Conference on Principles and Practice of Declarative Programming,
  27-29 August 2003, Uppsala, Sweden}}. \bibinfo{publisher}{{ACM}},
  \bibinfo{pages}{184--194}.
\newblock
\urldef\tempurl%
\url{https://doi.org/10.1145/888251.888269}
\showDOI{\tempurl}


\bibitem[\protect\citeauthoryear{Pareigis}{Pareigis}{1977}]%
        {Pareigis1977}
\bibfield{author}{\bibinfo{person}{Bodo Pareigis}.}
  \bibinfo{year}{1977}\natexlab{}.
\newblock \showarticletitle{Non-additive Ring and Module Theory II.
  C-categories, C-functors, and C-morphisms.}
\newblock \bibinfo{journal}{\emph{Publicationes Mathematicae Debrecen}}
  \bibinfo{volume}{24} (\bibinfo{year}{1977}), \bibinfo{pages}{351--361}.
\newblock


\bibitem[\protect\citeauthoryear{Plotkin}{Plotkin}{1977}]%
        {Plotkin77}
\bibfield{author}{\bibinfo{person}{G.D. Plotkin}.}
  \bibinfo{year}{1977}\natexlab{}.
\newblock \showarticletitle{LCF considered as a programming language}.
\newblock \bibinfo{journal}{\emph{Theoretical Computer Science}}
  \bibinfo{volume}{5}, \bibinfo{number}{3} (\bibinfo{year}{1977}),
  \bibinfo{pages}{223--255}.
\newblock
\showISSN{0304-3975}
\urldef\tempurl%
\url{https://doi.org/10.1016/0304-3975(77)90044-5}
\showDOI{\tempurl}


\bibitem[\protect\citeauthoryear{Polonowski}{Polonowski}{2013}]%
        {DBLP:conf/itp/Polonowski13}
\bibfield{author}{\bibinfo{person}{Emmanuel Polonowski}.}
  \bibinfo{year}{2013}\natexlab{}.
\newblock \showarticletitle{Automatically Generated Infrastructure for De
  Bruijn Syntaxes}. In \bibinfo{booktitle}{\emph{Interactive Theorem Proving -
  4th International Conference, {ITP} 2013, Rennes, France, July 22-26, 2013.
  Proceedings}} \emph{(\bibinfo{series}{Lecture Notes in Computer Science},
  Vol.~\bibinfo{volume}{7998})}, \bibfield{editor}{\bibinfo{person}{Sandrine
  Blazy}, \bibinfo{person}{Christine Paulin{-}Mohring}, {and}
  \bibinfo{person}{David Pichardie}} (Eds.). \bibinfo{publisher}{Springer},
  \bibinfo{pages}{402--417}.
\newblock
\urldef\tempurl%
\url{https://doi.org/10.1007/978-3-642-39634-2\_29}
\showDOI{\tempurl}


\bibitem[\protect\citeauthoryear{Pouillard}{Pouillard}{2012}]%
        {pouillard:tel-00759059}
\bibfield{author}{\bibinfo{person}{Nicolas Pouillard}.}
  \bibinfo{year}{2012}\natexlab{}.
\newblock \emph{\bibinfo{title}{{Namely, Painless: A unifying approach to safe
  programming with first-order syntax with binders}}}.
\newblock Theses. \bibinfo{school}{{Universit{\'e} Paris-Diderot - Paris VII}}.
\newblock
\urldef\tempurl%
\url{https://tel.archives-ouvertes.fr/tel-00759059}
\showURL{%
\tempurl}


\bibitem[\protect\citeauthoryear{Pouillard and Pottier}{Pouillard and
  Pottier}{2012}]%
        {nompa}
\bibfield{author}{\bibinfo{person}{Nicolas Pouillard} {and}
  \bibinfo{person}{Fran{\c{c}}ois Pottier}.} \bibinfo{year}{2012}\natexlab{}.
\newblock \showarticletitle{A unified treatment of syntax with binders}.
\newblock \bibinfo{journal}{\emph{J. Funct. Program.}} \bibinfo{volume}{22},
  \bibinfo{number}{4-5} (\bibinfo{year}{2012}), \bibinfo{pages}{614--704}.
\newblock
\urldef\tempurl%
\url{https://doi.org/10.1017/S0956796812000251}
\showDOI{\tempurl}


\bibitem[\protect\citeauthoryear{Sch{\"{a}}fer, Tebbi, and
  Smolka}{Sch{\"{a}}fer et~al\mbox{.}}{2015}]%
        {DBLP:conf/itp/SchaferTS15}
\bibfield{author}{\bibinfo{person}{Steven Sch{\"{a}}fer},
  \bibinfo{person}{Tobias Tebbi}, {and} \bibinfo{person}{Gert Smolka}.}
  \bibinfo{year}{2015}\natexlab{}.
\newblock \showarticletitle{Autosubst: Reasoning with de Bruijn Terms and
  Parallel Substitutions}. In \bibinfo{booktitle}{\emph{Interactive Theorem
  Proving - 6th International Conference, {ITP} 2015, Nanjing, China, August
  24-27, 2015, Proceedings}} \emph{(\bibinfo{series}{Lecture Notes in Computer
  Science}, Vol.~\bibinfo{volume}{9236})},
  \bibfield{editor}{\bibinfo{person}{Christian Urban} {and}
  \bibinfo{person}{Xingyuan Zhang}} (Eds.). \bibinfo{publisher}{Springer},
  \bibinfo{pages}{359--374}.
\newblock
\urldef\tempurl%
\url{https://doi.org/10.1007/978-3-319-22102-1\_24}
\showDOI{\tempurl}


\bibitem[\protect\citeauthoryear{Sewell, Nardelli, Owens, Peskine, Ridge,
  Sarkar, and Strnisa}{Sewell et~al\mbox{.}}{2010}]%
        {DBLP:journals/jfp/SewellNOPRSS10}
\bibfield{author}{\bibinfo{person}{Peter Sewell},
  \bibinfo{person}{Francesco~Zappa Nardelli}, \bibinfo{person}{Scott Owens},
  \bibinfo{person}{Gilles Peskine}, \bibinfo{person}{Thomas Ridge},
  \bibinfo{person}{Susmit Sarkar}, {and} \bibinfo{person}{Rok Strnisa}.}
  \bibinfo{year}{2010}\natexlab{}.
\newblock \showarticletitle{Ott: Effective tool support for the working
  semanticist}.
\newblock \bibinfo{journal}{\emph{J. Funct. Program.}} \bibinfo{volume}{20},
  \bibinfo{number}{1} (\bibinfo{year}{2010}), \bibinfo{pages}{71--122}.
\newblock
\urldef\tempurl%
\url{https://doi.org/10.1017/S0956796809990293}
\showDOI{\tempurl}


\bibitem[\protect\citeauthoryear{Stark, Sch{\"{a}}fer, and Kaiser}{Stark
  et~al\mbox{.}}{2019}]%
        {DBLP:conf/cpp/StarkSK19}
\bibfield{author}{\bibinfo{person}{Kathrin Stark}, \bibinfo{person}{Steven
  Sch{\"{a}}fer}, {and} \bibinfo{person}{Jonas Kaiser}.}
  \bibinfo{year}{2019}\natexlab{}.
\newblock \showarticletitle{Autosubst 2: reasoning with multi-sorted de Bruijn
  terms and vector substitutions}. In \bibinfo{booktitle}{\emph{Proceedings of
  the 8th {ACM} {SIGPLAN} International Conference on Certified Programs and
  Proofs, {CPP} 2019, Cascais, Portugal, January 14-15, 2019}},
  \bibfield{editor}{\bibinfo{person}{Assia Mahboubi} {and}
  \bibinfo{person}{Magnus~O. Myreen}} (Eds.). \bibinfo{publisher}{{ACM}},
  \bibinfo{pages}{166--180}.
\newblock
\urldef\tempurl%
\url{https://doi.org/10.1145/3293880.3294101}
\showDOI{\tempurl}


\bibitem[\protect\citeauthoryear{Streicher}{Streicher}{1991}]%
        {DBLP:books/daglib/0067012}
\bibfield{author}{\bibinfo{person}{Thomas Streicher}.}
  \bibinfo{year}{1991}\natexlab{}.
\newblock \bibinfo{booktitle}{\emph{Semantics of type theory - correctness,
  completeness and independence results}}.
\newblock \bibinfo{publisher}{Birkh{\"{a}}user}.
\newblock
\showISBNx{978-0-8176-3594-7}


\bibitem[\protect\citeauthoryear{Tanaka and Power}{Tanaka and Power}{2005}]%
        {Tanaka:2005:UCF:1088454.1088457}
\bibfield{author}{\bibinfo{person}{Miki Tanaka} {and} \bibinfo{person}{John
  Power}.} \bibinfo{year}{2005}\natexlab{}.
\newblock \showarticletitle{{A unified category-theoretic formulation of typed
  binding signatures}}. In \bibinfo{booktitle}{\emph{Proceedings of the 3rd ACM
  SIGPLAN workshop on Mechanized reasoning about languages with variable
  binding}} (Tallinn, Estonia) \emph{(\bibinfo{series}{MERLIN '05})}.
  \bibinfo{publisher}{ACM}, \bibinfo{address}{New York, NY, USA},
  \bibinfo{pages}{13--24}.
\newblock
\showISBNx{1-59593-072-8}
\urldef\tempurl%
\url{https://doi.org/10.1145/1088454.1088457}
\showDOI{\tempurl}


\bibitem[\protect\citeauthoryear{Team}{Team}{2021}]%
        {Coq}
\bibfield{author}{\bibinfo{person}{The Coq~Development Team}.}
  \bibinfo{year}{2021}\natexlab{}.
\newblock \bibinfo{title}{The {Coq} Proof Assistant, version 8.13.0}.
\newblock
\newblock
\urldef\tempurl%
\url{https://doi.org/10.5281/zenodo.4501022}
\showDOI{\tempurl}


\bibitem[\protect\citeauthoryear{{Univalent Foundations Program}}{{Univalent
  Foundations Program}}{2013}]%
        {hottbook}
\bibfield{author}{\bibinfo{person}{The {Univalent Foundations Program}}.}
  \bibinfo{year}{2013}\natexlab{}.
\newblock \bibinfo{booktitle}{\emph{Homotopy Type Theory: Univalent Foundations
  of Mathematics}}.
\newblock \bibinfo{publisher}{\url{http://homotopytypetheory.org/book}},
  \bibinfo{address}{Institute for Advanced Study}.
\newblock


\bibitem[\protect\citeauthoryear{Voevodsky}{Voevodsky}{2015}]%
        {DBLP:journals/mscs/Voevodsky15}
\bibfield{author}{\bibinfo{person}{Vladimir Voevodsky}.}
  \bibinfo{year}{2015}\natexlab{}.
\newblock \showarticletitle{An experimental library of formalized Mathematics
  based on the univalent foundations}.
\newblock \bibinfo{journal}{\emph{Math. Struct. Comput. Sci.}}
  \bibinfo{volume}{25}, \bibinfo{number}{5} (\bibinfo{year}{2015}),
  \bibinfo{pages}{1278--1294}.
\newblock
\urldef\tempurl%
\url{https://doi.org/10.1017/S0960129514000577}
\showDOI{\tempurl}


\bibitem[\protect\citeauthoryear{Voevodsky}{Voevodsky}{2016}]%
        {c-sys-from-relative-module}
\bibfield{author}{\bibinfo{person}{Vladimir Voevodsky}.}
  \bibinfo{year}{2016}\natexlab{}.
\newblock \showarticletitle{{C-system of a module over a $Jf$-relative monad}}.
\newblock  (\bibinfo{year}{2016}).
\newblock
\newblock
\shownote{\url{https://arxiv.org/abs/1602.00352}}.


\bibitem[\protect\citeauthoryear{Voevodsky, Ahrens, Grayson,
  et~al\mbox{.}}{Voevodsky et~al\mbox{.}}{2021}]%
        {UniMath}
\bibfield{author}{\bibinfo{person}{Vladimir Voevodsky},
  \bibinfo{person}{Benedikt Ahrens}, \bibinfo{person}{Daniel Grayson},
  {et~al\mbox{.}}} \bibinfo{year}{2021}\natexlab{}.
\newblock \bibinfo{title}{{UniMath --- a computer-checked library of univalent
  mathematics}}.
\newblock \bibinfo{howpublished}{{Available at
  \url{http://unimath.github.io/UniMath/} }}.
\newblock


\bibitem[\protect\citeauthoryear{Weirich and Aydemir}{Weirich and
  Aydemir}{2010}]%
        {lngen}
\bibfield{author}{\bibinfo{person}{Stephanie Weirich} {and}
  \bibinfo{person}{Brian Aydemir}.} \bibinfo{year}{2010}\natexlab{}.
\newblock \bibinfo{booktitle}{\emph{LNgen: Tool Support for Locally Nameless
  Representations}}.
\newblock \bibinfo{type}{{T}echnical {R}eport}.
\newblock
\urldef\tempurl%
\url{https://repository.upenn.edu/cis_reports/933/}
\showURL{%
\tempurl}


\bibitem[\protect\citeauthoryear{Zsid{\'o}}{Zsid{\'o}}{2010}]%
        {ju_phd}
\bibfield{author}{\bibinfo{person}{Julianna Zsid{\'o}}.}
  \bibinfo{year}{2010}\natexlab{}.
\newblock \emph{\bibinfo{title}{{Typed Abstract Syntax}}}.
\newblock \bibinfo{thesistype}{Ph.\,D. Dissertation}.
  \bibinfo{school}{University of Nice, France}.
\newblock
\newblock
\shownote{\url{http://tel.archives-ouvertes.fr/tel-00535944/}}.


\end{thebibliography}

\end{document}
\endinput
